\documentclass[english]{article}
\usepackage{custom_tex}
\usepackage{algorithm}
\usepackage{xcolor} 
\usepackage[noend]{algpseudocode}

\newcommand{\ce}{\mathcal{E}}
\newcommand{\eps}{\epsilon}     

\newcommand{\upper}{\mathcal{U}}  

\renewcommand{\top}{\mathsf{T}}
\newcommand{\esssup}{\mathrm{ess\,sup}}

\newcommand{\tod}{\Rightarrow}       
\newcommand{\esd}{\mathrm{ESD}}   

\newcommand{\cplx}{\mathbb{C}}
\newcommand{\im}{\mathrm{Im}}      

\newcommand{\mrd}{\,\mathrm{d}}   

\newcommand{\perm}{\mathrm{perm}} 

\title{Deterministic parallel analysis: An improved method for selecting the number of factors and principal components}  

\date{\today}
\author{Edgar Dobriban\footnote{Department of Statistics, The Wharton School, University of Pennsylvania. E-mail: \texttt{dobriban@wharton.upenn.edu}.} \, and Art B. Owen\footnote{Department of Statistics, Stanford University. 
E-mail: \texttt{owen@stanford.edu}.}}

\begin{document}
\maketitle
\begin{abstract}
Factor analysis and principal component analysis (PCA) are used in many application areas.
The first step, choosing the number of components, remains a serious challenge. Our work proposes improved methods for this important problem. 
One of the most popular state-of-the-art methods is \emph{parallel analysis} (PA), which compares the observed factor strengths to simulated ones under
a noise-only model.
This paper proposes improvements to PA. We first derandomize it, proposing \emph{Deterministic Parallel Analysis} (DPA), which is faster and more reproducible than PA.
Both PA and DPA are prone to a \emph{shadowing} phenomenon in
which a strong factor makes it hard to
detect smaller but more interesting factors.
We propose \emph{deflation} to counter shadowing.
We also propose to raise the decision threshold to improve estimation accuracy. We prove several consistency results for our methods, and test them in simulations. 
We also illustrate our methods on data from the Human Genome Diversity Project (HGDP), where they significantly improve the accuracy. 
\end{abstract}

\section{Introduction}

Factor analysis is widely used in many application areas \cite[see e.g.,][]{malinowski2002factor, bai2008large, brown2014confirmatory}.
Given data describing $p$ measurements made on $n$ objects or $n$ subjects, factor analysis summarizes it
in terms of some number $k\ll\min(n,p)$ of latent factors,
plus noise. Finding and interpreting those factors may be very useful.
However, it is hard to choose the number $k$ of factors to 
include in the model.

There have been many proposed methods to choose the number $k$ of factors.
Classical methods include using an a priori fixed singular value threshold
\citep{kaiser1960application} and the scree test, a subjective search for a gap
separating small from large singular values \citep{cattell1966scree}.
The recent surge in data sets with $p$ comparable to $n$ or even larger
than $n$ has been met with new methods 
in econometrics \citep{bai2012statistical,ahn2013eigenvalue} 
and random matrix theory \citep{nadakuditi2008sample}
among many others.  In this paper we focus exclusively on the
method of parallel analysis (PA), first introduced by \cite{horn:1965}
and variations of it.

PA, especially the version by \cite{buja:eyob:1992},
 is one of the most popular ways to pick $k$ in factor analysis.
It is used, for example as a step in bioinformatics 
papers such  as \cite{leek2008general}.
It seems that the popularity of PA is because users find it `works well in practice'.
That property is hard to quantify on data with an unknown ground truth. 
We believe PA `works well' because it selects factors above
the \emph{largest noise singular value} of the data, as shown by \cite{dobriban2017factor}.
Our goal in this paper is to provide an improved method that estimates
this threshold faster and more reproducibly than PA does.

We will describe and survey parallel analysis in detail below.
Here we explain how it is based on sequential testing, assuming some
familiarity with factor analysis.
If there were no factors underlying the data, then the $p$ variables
would be independent.
We can then compare the size of the apparent
largest factor in the data matrix (the largest singular value) 
to what we would see in data with
$n$ IID observations, each of which consists of $p$ independent but
not identically distributed variables. 
We can simulate such data sets many times, either by generating new random variables or by permuting each column of the original data independently. If the first singular value
is above some quantile of the simulated first singular values, then
we conclude that $k\ge1$.
If we decide that $k\ge1$, then we check whether the second factor (second singular value)
has significantly large size compared to simulated second factors, and so on.
If the $k+1$'st factor is not large compared to simulated ones, then we stop with $k$ factors.

Our first improvement on PA is to derandomize it. 
For large $n$ and $p$, the size of the largest singular value
can be predicted very well using tools from 
random matrix theory.
The prediction depends very little on the distribution of the
data values, a fact called \emph{universality}.  For example, random matrix theory developed
for Gaussian data \citep{johnstone2001distribution}
has been seen to work very well on single
nucleotide polymorphism (SNP) data that take values
in $\{0,1,2\}$ \citep{patterson2006population}.
The recently developed Spectrode tool \citep{dobriban2015efficient}
computes the \emph{canonical distribution of noise singular values}---the so-called Silverstein-Marchenko-Pastur distribution---without any simulations. We propose to use the upper edge of that
distribution as a factor selection threshold.

We call the resulting method DPA for deterministic PA.
A deterministic threshold has two immediate advantages over PA.
It is \emph{reproducible}, requiring no randomization. It is also \emph{much faster} because
we do not need to repeat the factor analysis on multiple randomly 
generated data sets.  

DPA also \emph{frees the user from choosing an ad hoc
significance level} for singular values. 
If there are no factors in the data, then PA will mistakenly
choose $k\ge1$ with a probability corresponding to the quantile
used in simulation, because the top sample singular value has
the same distribution as simulated ones.  As a result, we need a huge number of simulations to get sufficiently extreme quantiles that we rarely make mistakes. By contrast, DPA
can be adjusted so that the probability of false positives vanishes.
For instance the user can insist that a potential factor
must be $1.05$ times larger than what we would see in uncorrelated noise.

\cite{zhou2017eigenvalue} recently developed a related method to derandomize PA, also by employing \cite{dobriban2015efficient}'s Spectrode algorithm to approximate the noise level. However, there are several crucial differences between our approaches: First, we use the variances of the columns as the population spectral distribution to be passed to Spectrode, while \cite{zhou2017eigenvalue} fit a five-point mixture to the empirical eigenvalue spectrum, and use the resulting upper edge. Second, we provide theoretical consistency results, while \cite{zhou2017eigenvalue} evaluate the methods empirically on some genetic data. 

In addition to randomness,  PA has a problem with \emph{shadowing}.  
If the first factor is very large then the simulations behind PA
lump that factor into the noise, raising the threshold for the next factors. Then PA can miss some of the next factors.
This was seen empirically by
\cite{owen:wang:2016} and explained theoretically by \cite{dobriban2017factor} who coined the term `shadowing'.
It is fairly common for real data sets to contain one `giant factor' explaining
a large proportion of the total variance.  In finance, many stock returns
move in parallel.  In educational testing, the `overall ability' factor 
of students is commonly large.  
Such giant factors are usually well known a priori
and can shadow smaller factors that might 
otherwise have yielded interesting discoveries.

We counter shadowing by \emph{deflation}. If the first factor is large enough
to keep, then we estimate its contribution and subtract it from the data.
We then look for additional factors within the residual.
The resulting deflated DPA (DDPA) is less affected by shadowing.

Deflating the $k$'th factor counters shadowing by lowering the threshold for the $k+1$'st. The result can be a cascade where threshold reductions by DDPA
trigger increased $k$ and vice versa.  
Recent work has shown that the threshold for detecting a factor is lower
than the threshold for getting beneficial estimates of the corresponding
eigenvectors \citep{perr:2009,gavish2014optimal}.
Because deflation involves subtracting an outer product
of weighted eigenvectors, we consider raising the deflation
threshold to only keep factors that improve estimation accuracy.
We call this method DDPA+.

\cite{vitale2017selecting} proposed a method for selecting the number of components using permutations as in PA, including deflation, using a standardized eigenvalue test statistic, and adjusting the permuted residual by a certain projection method. Their first two steps are essentially deflated (but not derandomized) PA. Similarly to  DDPA, this suffers from over-selection, and their remaining two steps deal with that problem. Their approach differs from ours in several important aspects, including that our final method is deterministic, and that we prove rigorous theoretical results supporting it.

An outline of this paper is as follows.
Section~\ref{sec:background} gives some context
on factor analysis, random matrices, and parallel analysis.
Section~\ref{sec:derandomize} presents our deterministic PA  (DPA).
We show that, similarly to PA, DPA selects \emph{perceptible} factors and omits \emph{imperceptible} factors. 
Section~\ref{sec:deflation} presents our deflated version of DPA, called DDPA
that counters the shadowing phenomenon by subtracting
estimated factors. We raise the threshold in DDPA to improve estimation accuracy. We call this final version DDPA+.
Section~\ref{sec:numerical} presents numerical simulations
where we know the ground truth. In our simulations, 
DPA is more stable than PA and much faster too. DDPA is better able to counter shadowing than
DPA but the number of factors it chooses has more variability than
with DPA.  Raising the decision threshold in DDPA+ appears to reduce its variance.
The results are reproducible with software available from \url{github.com/dobriban/DPA}. 
In Section~\ref{sec:hgdp}
we compare the various methods on the Human Genome Diversity Project (HGDP) dataset \citep{li2008worldwide}.  
We believe that PA, DPA, and DDPA choose too many
factors in that data. The choice for DDPA+ seems to be more accurate but conservative.
Section~\ref{sec:getthresh} gives a faster way to
compute the DPA threshold we need by avoiding the full Spectrode computation
as well as a way to compute the DDPA+ threshold.
We close with a discussion of future work in Section \ref{sec:disc}. 
The proofs are presented in Section \ref{sec:supp}. 

\section{Background}\label{sec:background}

In this section we introduce our notation and quote some results from random matrix theory.  
We also describe goals of factor analysis and survey some of the literature on parallel analysis. Space does not allow a full account of factor analysis.

\subsection{Notation}

The data are the $n\times p$ matrix $X$ and contain $p$ features of $n$ samples.  The rows of $X$ are IID $p$-dimensional observations $x_i$, for $i=1,\dots,n$. We will work with the standard factor model, where observations have the form $x_i=\Lambda\eta_i+\ep_i$. Here $\Lambda$ is a non-random ${p\times r}$ matrix of factor loadings, $\eta_i$ is the $r$-dimensional vector of factor scores of the $i$-th observation, and $\ep_i$ the $p$-dimensional noise vector of the $i$-th observation.

In matrix form, this can also be written as 
\begin{align}\label{eq:factormodel}
X =\eta\Lambda^\top +Z\Phi^{1/2}.
\end{align}
Here, $\eta$ is an $n \times r$  random matrix containing the factor values, so
$\eta_i$ is the $i$'th row of $\eta$. Also  $\ce=Z\Phi^{1/2}$ is the $n \times p$ matrix containing the noise, so  $\ep_i$ is the $i$'th row of $\ce$. 

Thus, the noise is modeled as $\ep_i = \Phi^{1/2} z_i$, where $z_i$ is the $i$'th row of $Z$. 
The $n\times p$ matrix $Z$ contains independent random variables with mean $0$
and variance $1$, and $\Phi$ is a $p\times p$ diagonal matrix
of idiosyncratic variances $\Phi_j$, $j=1,\dots,p$.
The covariance matrix of $x_i$
is $\Sigma = \Lambda\Psi\Lambda^\top+\Phi$
where $\Psi\in\real^{r\times r}$ is the covariance of $\eta_i$.
We define the scaled factor loading matrix as
$\Lambda\Psi^{1/2}\in\real^{p\times r}$. Note that we could reparametrize the problem by this  matrix. The factor model has known identifiability problems, see for instance \cite{bai2012statistical} 
for a set of possible resolutions. However, in our case, we only want to estimate the number of large factors, which is asymptotically identifiable. 

The model~\eqref{eq:factormodel} contains $r$ factors, and we generally do not know $r$. We use $k$ factors and $k$ is not necessarily the same as $r$.   


We will need several matrix and vector norms.
For a vector $v\in\real^m$ we use the usual $\ell_1$, $\ell_2$, and $\ell_\infty$
norms $\Vert v\Vert_1$, $\Vert v\Vert_2$ and $\Vert v\Vert_\infty$.
An unsubscripted vector norm $\Vert v\Vert$ denotes $\Vert v\Vert_2$.
For a matrix $A\in\real^{n\times r}$ we use the induced norms 
$$\Vert A\Vert_p=\sup_{v\ne0}\Vert Av\Vert_p/\Vert v\Vert_p.$$
The spectral---or operator---norm is  $\Vert A\Vert_2$,
while $\Vert A\Vert_\infty =\max_{1\le i\le n}\sum_{j=1}^r|X_{ij}|$
and $\Vert A\Vert_1 =\max_{1\le j\le r}\sum_{i=1}^n|X_{ij}|$.
The Frobenius norm is given by $\Vert A\Vert_F^2=\sum_i\sum_jA_{ij}^2$.

We consider $\eta\Lambda^\top$ to be the signal matrix and
$\ce = Z\Phi^{1/2}$ to be the noise.
It is convenient to normalize the data,  via
\begin{align}\label{eq:normalized}
n^{-1/2}X=S+N\quad\text{with}\quad S=n^{-1/2}\eta\Lambda^\top
\quad\text{and}\quad N=n^{-1/2}\ce.
\end{align}

The aspect ratio of $X$ is $\gamma_p = p/n$.
We consider limits with $p\to\infty$ and 
$\gamma_p\to\gamma\in(0,\infty)$. In our models, under some precise assumptions stated later, $\Vert N\Vert_2\to b>0$ almost surely, for some $b\in(0,\infty)$.
We call $b$ the \emph{size of the noise}.

For a positive semi-definite matrix $\Sigma\in\real^{p\times p}$, with eigenvalues
$\lambda_1\ge\lambda_2\ge\lambda_p\ge0$, its empirical spectral distribution is 
$$
\esd(\Sigma) = \frac1p\sum_{j=1}^p\delta_{\lambda_j}
$$
where $\delta_\lambda$ denotes a point mass distribution at $\lambda$.
In our asymptotic setting, $\esd(\Sigma)$ converges weakly to
a distribution $H$, denoted $\esd(\Sigma)\tod H$.
For a general bounded probability distribution $H$, its 
\emph{upper edge} is the essential supremum 
$$\upper(H)\equiv\esssup(H) =\inf\bigl\{ M\in\real\mid H((-\infty,M])=1\bigr\} .$$ 

The sample covariance matrix is $\hat\Sigma =(1/n)X^\top X$, and
we let $\hat D = \diag (\hat\Sigma)$.
If there are no factors, then $\hat D$ is a much better estimate of $\Sigma$
than $\hat\Sigma$ is.
We call $\esd(\hat D)$ the \emph{variance distribution}. 
The singular values of a matrix $X\in\real^{n\times p}$ are denoted by
$\sigma_1(X)\ge\sigma_2(X)\ge\cdots\ge\sigma_{\min(n,p)}(X)\ge0$.
The eigenvalues of $\hat\Sigma$ are $\lambda_j = \sigma^2_j(n^{-1/2}X)$.

\subsection{Random matrix theory}

Here we outline some basic results needed  from random matrix theory, see \citep{bai2009spectral, paul2014random,yao2015large} for reference. We focus on the Marchenko-Pastur distribution and the Tracy-Widom distribution.

For the spherically symmetric case
$\cov(x_i)=I_p$ we have $\esd(I_p)=\delta_1$,
but the eigenvalues of the sample covariance matrix $\hat\Sigma=(1/n)X^\top X$
do not all converge to unity in our limit when the sample size is only a constant times 
larger than the dimension.
Instead the empirical spectral distrbution of $\hSigma$ converges to the \emph{Marchenko-Pastur (MP) distribution}, $\esd(\hat\Sigma)\tod F_\gamma$ \citep{marchenko1967distribution}.  
If $n$ and $p$ tend to infinity with $p/n\to\gamma\in(0,1]$
then $F_\gamma$ has probability density function
\begin{align}\label{eq:mpdist}
f_\gamma(\lambda) &=
\begin{cases}
\dfrac{\sqrt{(\lambda-a_\gamma) (b_\gamma-\lambda)}}
{2\pi\lambda\gamma}, & a_\gamma\le \lambda\le b_\gamma\\[2ex]
0, &\text{else,}
\end{cases}
\end{align}
where
$$a_\gamma = (1-\sqrt{\gamma})^2\quad\text{and}\quad
b_\gamma = (1+\sqrt{\gamma})^2.
$$
If $\gamma>1$, then $F_\gamma$ is a mixture placing
probability $1/\gamma$ on the above density
and probability $1-1/\gamma$ on a point mass at $0$.
In either case, it is also known that the largest eigenvalue of the sample covariance matrix converges to the upper edge of the MP distribution
$\upper(\esd(\hat\Sigma))\to (1+p/n)^2>\upper(\esd(\Sigma))=1$.
If the components of $\ep_i$ are IID with kurtosis $\kappa\in(-2,\infty)$, then $\upper(\esd(\hat D))$
is the largest of $p$ IID random variables with mean $1$ and variance $(\kappa+2)/n$.
Then $\upper(\esd(\hat D))$ will be much closer to $1$ than $\upper(\esd(\hat\Sigma))$ is.

For a general distribution $H$ of bounded nonnegative eigenvalues of $\Sigma$,
there is a generalized Marchenko-Pastur distribution
$F_{\gamma,H}$
\citep{bai2009spectral, yao2015large}. 
If $\gamma_p\to\gamma$ and $\esd(\Sigma)\to H$,
then under some moment conditions $ \smash{\esd(\hat\Sigma)\to F_{\gamma,H}}$.

We can use Spectrode \citep{dobriban2015efficient} to compute $F_{\gamma,\hat H}$
and obtain $\upper( F_{\gamma,\hat H})$
for any $ \smash{\hat H}$.  Typically we will use $\smash{\hat H=\esd(\hat\Sigma)}$ or $ \smash{\esd(\hat D)}$.
It is convenient to use $ \smash{F_{\gamma,\hat\Sigma}}$ as a shorthand for
$ \smash{F_{\gamma,\esd(\hat\Sigma)}}$.
If the data are noise without any
factors, we can estimate $\Sigma$ by
the diagonal matrix of variances
$\hat D = \diag(\hat\Sigma)$ and expect
that the largest eigenvalue of the empirical noise covariance matrix $\hSigma$ should
be close to $\upper(F_{\gamma,\hat D})$.

The largest eigenvalue of the sample covariance matrix
is distributed around the upper edge of $F_{\gamma,H}$.
The difference, appropriately scaled,
has an asymptotic Tracy-Widom distribution.
See \cite{hachem2016large} and \cite{lee2016tracy} for recent results. 
In Gaussian white noise, \cite{johnstone2001distribution}
shows that the scale parameter is
\begin{align}\label{eq:imjtw}
\tau_p =n^{-1/2}\cdot (1+(p/n)^{1/2})\cdot(n^{-1/2}+p^{-1/2})^{1/3},
\end{align}
so that fluctuations in the largest eigenvalue quickly become negligible.

\subsection{Factor analysis}\label{sec:factor}

A common approach to factor analysis begins with a principal component analysis (PCA) of the data, and selects factors based on the size of the singular values of $n^{-1/2}X$. 
Let $n^{-1/2}X$ have singular values $\sigma_j$ for $j=1,\dots,\min(n,p)$.
Now choosing $k$ is equivalent to choosing a threshold $\sigma_*$ and retaining
all factors with $\sigma_\ell(n^{-1/2}X)>\sigma_*$.
There are several different goals for factor analysis and different numbers $k$ may be best for each of them. Here we list some goals.

In some cases the goal is to select \emph{interpretable} factors. In many applications, there are known factors that we expect to see in the data. These are the giant factors mentioned above, such as in finance, where many stock returns move in parallel. In such cases, interest centers on new interpretable factors beyond the already known ones. For instance, we may look for previously unknown segments of the market. The value of the threshold is what matters here, however, users will typically also inspect and interpret the factor loadings.

Another goal is \emph{estimation} of the signal matrix $\eta\Lambda^\top$, without necessarily
interpreting all of the factors, as a way of denoising $X$.
Then one can choose a threshold by calculating the effect of retained factors on the mean squared error. Methods for estimating the unobserved signal have been developed under various assumptions by \cite{perr:2009}, \cite{gavish2014optimal}, \cite{nadakuditi2014optshrink} and \cite{owen:wang:2016}.
The optimal $k$ for estimation can be much smaller than the true value $r$.

In bioinformatics, factor models are useful for \emph{decorrelating} genomic data so
that hypothesis test statistics are more nearly independent.
See for example, \cite{leek2007capturing}, \cite{sun2012}, \cite{gagnon2013removing},
and~\cite{gerard2017unifying}.
Test statistics that are more nearly independent are better suited to multiplicity
adjustments such as the Benjamini-Hochberg procedure. 


\subsection{Parallel analysis}

Parallel analysis was introduced by \cite{horn:1965}. The idea is to simulate
independent  data with the same variances as the original data but
without any of the correlations that factors would induce.
If the observed top eigenvalue of the sample covariance is significantly
large, say above the simulated 95'th percentile, then we include that
factor/eigenvector in the model. If the first $k\ge1$ factors have been
accepted, then we accept the $k+1$'st one if it is above the 95'th percentile
of its simulated counterparts. Horn's original approach used the mean instead of the 95'th percentile.
Horn was aware of the bias in the largest eigenvalue of sample covariance matrices even before 
\cite{marchenko1967distribution} characterized it theoretically.

There is a lot of empirical evidence in favor of PA. For instance,
\cite{zwic:veli:1986} compared five methods in some small
simulated data sets.  Their criterion was choosing a correct $k$
not signal recovery. They concluded that  ``\emph{The PA method was consistently accurate.}'' The errors of PA were typically to overestimate
$k$ (about 65\% of the errors). In addition, \cite{glor:1995} wrote that ``\emph{Numerous studies have consistently shown that Horn's
parallel analysis is the most nearly accurate methodology for
determining the number of factors to retain in an exploratory
factor analysis.}'' He also reported that the errors are mostly of overestimation.



The version of PA 
by \cite{buja:eyob:1992} replaces Gaussian simulations by independent 
random permutations within the columns of the data matrix (See Algorithm \ref{pa}).
Permutations have the advantage of removing the parametric assumption.
They are widely used for decorrelation in genomics to identify latent effects such
as batch variation.

\begin{algorithm}
\caption{\textsc{PA}: Parallel Analysis}
\begin{algorithmic}[1]
\State \textbf{input}: Data $X\in\real^{n\times p}$, centered, containing $p$ features of $n$ samples. Number of permutations $n_p$ (default = 19). Percentile $\alpha$ (default = 100). 
\State Generate $n_p$ independent random matrices $X_\pi^i$, $i=1,\ldots,n_p$ in the following way. Permute the entries in each column of $X$ (each feature) separately.
\State Initialize:  $k\gets0$. 
\State $H_k \gets$ distribution of the $k$-th largest singular values of $X_\pi^i$, $i=1,\ldots,n_p$. 
\If{$\sigma_k(X)$ is greater than the $\alpha$-th percentile of $H_k$}
\State $k\gets k+1$
\State  Return to step 4. 
\EndIf
\State \textbf{return}: Selected number of factors $k$. 
\end{algorithmic}
\label{pa}
\end{algorithm}


PA was developed as a purely empirical method, without theoretical justification. 
Recently, \cite{dobriban2017factor} developed a theoretical understanding: 
PA selects certain perceptible factors in high-dimensional factor models, as described below. \cite{dobriban2017factor} clarified the limitations of PA, including the shadowing phenomenon mentioned above. 

\section{Derandomization}\label{sec:derandomize}

If there are no factors, so $r=0$ in~\eqref{eq:factormodel}, then
$X=Z\Phi^{1/2}$ is a matrix of uncorrelated noise.
We can estimate $\Phi$ by $\hat D = \diag(X^\top X/n)$
and compute the upper edge of the MP distribution induced by $\hat D$ with aspect ratio $p/n$, i.e., $\upper( F_{p/n,\hat D})$. This is a deterministic approximation of the location of the largest noise eigenvalue of $n^{-1} X^\top X$.

Our new method of \emph{Derandomized Parallel Analysis} (DPA) chooses $k$ factors where $k$ satisfies
$$\sigma^2_k(n^{-1/2}X) > \upper(F_{p/n,\hat D}) \ge\sigma^2_{k+1}(n^{-1/2}X).$$
The Spectrode method \citep{dobriban2015efficient}
computes $F_{p/n,\hat D}$ from which
we can obtain this upper edge. 
We give an even faster algorithm in Section \ref{sec:fast_edge}
that computes $\upper(F_{p/n,\hat D})$ directly.

The threshold that DPA estimates is a little different from the PA one.
For the $k$'th singular value, PA is estimating something like
the $1-k/p$ quantile of the generalized Marchenko-Pastur distribution
while DPA uses the upper edge.   This difference is minor for small $k$.

We say that factor $k\in\{1,\dots,r\}$
is \emph{perceptible}
if $\sigma_k(n^{-1/2}X)>b+\eps$ a.s.\ for some $\eps>0$,
where $b$ is the almost sure limit of $n^{-1/2}\Vert \ce\Vert_2$
defined in Section~\ref{sec:background}.
It is \emph{imperceptible}
if $\sigma_k(n^{-1/2}X)<b-\eps$ a.s.\ holds.  

Perceptible factors are defined in terms of $X$ and $N$, and not in terms of underlying factor
signal size. This may seem a bit circular, or even tautological. However, there are two reasons for adopting this definition. First, as \cite{dobriban2017factor} has argued, in certain spiked covariance models one can show that a large underlying signal size leads to separated factors. Second, this is the ``right'' definition mathematically, leading to elegant proofs. Moreover, the definition is related to the BBP phase transition \citep{baik2005phase}; and it reduces to the BBP phase transition for the usual spiked models. However, our definition also makes sense in other models, where $p/n\to\infty$ and the number of spikes diverges to infinity. These are not included in usual spiked models. See \cite{dobriban2017factor} Section 5, for a detailed explanation. 


Theorem~\ref{thm:fa_cons} below
shows that DPA selects all perceptible factors
and no imperceptible factors, just like PA does. While the result is similar to those on PA from \cite{dobriban2017factor}, the proof is different. 
Let $\Psi^{1/2}$ be the symmetric square root of $\Psi=\cov(\eta_i)$, and recall that we defined the scaled factor loading matrix $\Lambda\Psi^{1/2}$. 

\begin{theorem}[DPA selects perceptible factors]
\label{thm:fa_cons}
Let the centered data matrix $X$ follow the factor model~\eqref{eq:factormodel} with 
IID $x_i = \Lambda \eta_i +\ep_i$, $i=1,\dots,n$. Assume
the following conditions (for some positive constants $\eps$ and $\delta$):
\begin{compactenum}[\quad 1)]
\item {\bf Factors}: The factors $\eta_i$ have the form $\eta_i  = \Psi^{1/2} U_i$, where $U_i$ has $r$ independent  entries with mean zero, variance one,  and bounded moments of order $4+\delta$.
\item {\bf Idiosyncratic terms}:  The idiosyncratic terms are $\ep_i=\Phi^{1/2} Z_i$, where $\Phi^{1/2}$ is a diagonal matrix, and $Z_i$ have $p$ independent entries with mean zero, variance one, and bounded moments of order $8+\eps$. 
\item  {\bf Asymptotics}: $n,p \to \infty$, such that $p/n \to \gamma >0$, and
there is a bounded distribution $H$ with
$\esd(\Phi)\tod H$ and $\max_{1\le j\le p} \Phi_j \to \upper(H)$. 
\item {\bf  Number of factors}: 
The number $r$ of factors is fixed, or grows such that $r/n^{1+\delta/4}$ is summable. 
\item {\bf  Factor loadings}: The  scaled loadings are delocalized, in the sense that $\Vert\Lambda\Psi^{1/2}\Vert_\infty\to0$. 
\end{compactenum} 
Then with probability tending to one, DPA  \emph{selects all perceptible factors}, and \emph{no imperceptible factors}. 
\end{theorem}
\begin{proof}
See Section~\ref{fa_cons_pf}. 
\end{proof}

This shows that the behavior of DPA is similar to that of PA \citep{dobriban2017factor}. However, there are some important differences. In contrast to the result on PA, Theorem~\ref{thm:fa_cons}
allows both a \emph{growing number} of factors, as well as \emph{growing factor strength}. Specifically, the number of factors can grow, as long as $r/n^{1+\delta/4}$ is summable. For instance, if we assume that $\delta=4$, which translates into the condition that the 8-th moments of the entries of $U_i$ are uniformly bounded, then $r$ can grow as fast as $n^{1-\eps}$ for any $\eps>0$. 

The factor strengths can grow as follows.
Let the scaled factor loading matrix $\Lambda\Psi^{1/2}$  have columns $d_\ell$
for $\ell =1,\dots,r$.
These columns can have a diverging Euclidean norm $\Vert d_\ell\Vert_2\to\infty$
so long as the sum of the max-norms tends to zero, $\sum_{\ell=1}^r\Vert d_\ell\Vert_\infty \to 0$.
As a result, the Frobenius norm of the scaled loading matrix can
grow,  subject to $\Vert\Lambda\Psi^{1/2}\Vert_{F}=o([n/r]^{1/2})$. 
There is a tradeoff between the number of factors and their sizes. 
For instance, if $r \sim n^{1-\eps}$, then the scaled factors can grow
while subject to their Euclidean norm being  $\Vert d_\ell\Vert = o(n^{\eps})$. 

DPA may yield false positives and false negatives stemming from factors
that are not perceptible.  
In particular, there is a positive probability,
even asymptotically, for DPA to pick $k\ge1$ from data with no factors. 
This may be why \cite{glor:1995} found PA's errors were ones of overestimation. 

We can remove false positives asymptotically by only taking 
factors above some multiple $1+\eps_p$ of the estimated upper edge. 
That is, we select a factor if $\sigma_k>(1+\eps_p) \cdot \upper(F_{p/n,\hat D})^{1/2}$.
The choice of $\eps_p$ is problem dependent; we think $0.05$ is large enough to trim out noise in the bioinformatics problems we have in mind.  A factor that is only $1.05$ times as large as what we would expect in pure noise is not likely to be worth interpreting. 

%

For estimation of the number of factors,
we might want to take $\eps_p\to0$ in a principled way
using the Tracy-Widom distribution.
We could take $\eps_p =$ $\smash{c_{n,p}\tau_p^{1/2}}$, where $c_{n,p}$ is a deterministic sequence
and $\tau_p$ is given by~\eqref{eq:imjtw}. The constant $c_{n,p}$ can be chosen as some fixed percentile of the Tracy-Widom distribution. However, this leads to only a tiny increase in the threshold, so we decided it is not worth making the method more complicated with this additional parameter.  Therefore, in the following DDPA, DDPA+ algorithms and simulations, $\eps_p=0$.

For estimation of the signal matrix $\eta\Lambda^\top$, we will discuss a different choice in Section \ref{evec_estim}.



\section{Deflation}\label{sec:deflation}

To address the shadowing of mid-sized factors, we will deflate $X$.
Let the singular value decomposition (SVD) of $X \in\real^{n\times p}$
be $X = \sum_{i=1}^{\min(n,p)} \sigma_i u_i v_i^\top$, with non-increasing $\sigma_i$.
If we decide to keep the first factor, then we subtract $\sigma_1u_1v_1^\top$
from $X$, recompute the upper edge of the MP distribution and repeat the process.
This ensures that the very strong factors are removed, and thus they do not shadow mid-sized factors. 
We call this DDPA for deflated DPA.



\begin{algorithm}
\caption{\textsc{DDPA}: Deflated Deterministic Parallel Analysis}
\begin{algorithmic}[1]
\State \textbf{input}: Data $X\in\real^{n\times p}$, centered, containing $p$ features of $n$ samples 
\State Initialize:  $k\gets0$. 
\State Compute variance distribution: $\hat H_p \gets\esd(\diag(n^{-1}X^\top X))$. 
\If{$\sigma_1(n^{-1/2}X)>(1+\eps_p)\upper(F_{\gamma_p, \hat  H_p})^{1/2}$,   [by default $\eps_p=0$]} \label{item:ddpathresh} 
\State $k\gets k+1$
\State $X \gets X - \sigma_1 u_1 v_1^\top$ (from the SVD of $X$) 
\State  Return to step 3. 
\EndIf
\State \textbf{return}: Selected number of factors $k$. 
\end{algorithmic}
\label{ddpa}
\end{algorithm}

After $k$ deflations, the current residual matrix  is $X=\sum_{i=k+1}^{\min{n,p}}\sigma_i u_i v_i^\top$.
Thus, this algorithm requires only one SVD computation. 

Deflation is a \emph{myopic} algorithm.  At each step we inspect the strongest apparent factor.
If it is large compared to a null model, we include it. Otherwise we stop. There is no lookahead.
There is a potential downside to myopic selection.  If some large singular values are close to each
other, but well separated from the bulk, the algorithm might stop without taking any of them,
even though it could be better to take all of them.  We think that this sort of shadowing from below
is unlikely in practice.

Deflation must be analyzed carefully, as the errors in estimating singular values and vectors
could propagate through the algorithm.
To develop a theoretical understanding of deflation, our first step is to show that it works in a setting which parallels Theorem \ref{thm:fa_cons}. Recall that the scaled factor loading matrix is $\Lambda\Psi^{1/2} = [d_1,\ldots,d_r]$.  The statement of this theorem involves the Stieltjes transform $m_{\gamma,H} = \E_H 1/(X-z)$ which is defined in more detail in the supplement. 

\begin{theorem}[DDPA selects the perceptible factors]
\label{thm:defpa_cons}
Consider factor models under the conditions 1 through 3 of Theorem \ref{thm:fa_cons},
with modified conditions 4 and 5 as follows:
\begin{compactenum}[\quad 1)]
\setcounter{enumi}{3}
\item {\bf  Number of factors}: 
The number $r$ of factors is fixed. 
\item {\bf  Factor loadings}: The vectors of scaled loadings $d_\ell$ are delocalized in the sense that, $\Vert d_\ell\Vert_\infty \to 0$ for $\ell=1,\dots,r$. They are also delocalized with respect to $\Phi$, in the sense that 
\begin{align}\label{eq:alsodelocalized}
\frac{x^\top (\Phi -zI_p)^{-1} d_\ell - m_{\gamma,H}(z) \cdot x^\top d_\ell}{\Vert d_\ell\Vert} \to 0
\end{align}
uniformly for $\Vert x\Vert\le1$, $\ell=1,\dots,r$, and $z\in\cplx$ with $\im(z) >0$ fixed.
\end{compactenum} 
Then with probability tending to one, DDPA with $\ep_p=0$
\emph{selects all perceptible factors}, and \emph{no imperceptible factors}.
\end{theorem}
\begin{proof}
See the supplement. 
\end{proof}

The proof requires one new key step, controlling the $\ell_\infty$ norm of the empirical spike eigenvectors in certain spiked models.  The delocalization condition~\eqref{eq:alsodelocalized} is somewhat technical, and requires that $d_\ell$ are ``random-like'' with respect to $\Phi$. 
This condition holds if $d_\ell$ are random vectors with IID coordinates, albeit only in the special case where the entries of $\Phi$ have a vanishing variability, so that we are in a near-homoskedastic situation. For other $\Phi$ matrices, there can of course be other vectors $d_\ell$ that make this condition hold. 

In practice, the advantage of deflation is that it works much better than plain DPA in the ``explosive'' setting with strong factors. We see this clearly in the simulations of Section~\ref{sec:numerical}.
Our theoretical results for deflation are comparable to usual DPA, so this empirical advantage is not reflected in Theorem \ref{thm:defpa_cons}. Analyzing  the strong-factor model theoretically seems to be beyond what the current proof techniques can accomplish.
 
%
As with DPA, we might want to increase the threshold in order to trim factors that are
not perceptible, or too small to be interesting or useful.

\begin{proposition}[Slightly increasing the threshold retains the properties of DDPA]
Consider the DDPA algorithm above where 
in step~\ref{item:ddpathresh} of algorithm \ref{ddpa} we select if $\sigma_1(n^{-1/2}X)>(1+\eps_p)\upper(F_{\gamma_p, \hat  H_p})^{1/2}$, and $\eps_p \to 0$.
Then the resulting algorithm satisfies Theorem \ref{thm:defpa_cons}.
\end{proposition}
\begin{proof}
The proof is immediate, by checking that the steps of the proof of Theorem \ref{thm:defpa_cons} go through.
\end{proof}

\subsection{DDPA+: Increasing the threshold to account for eigenvector estimation}
\label{evec_estim}

Every time deflation removes a factor, the threshold for the next factor is decreased.
For some data, DDPA will remove many factors one by one.  It is worth considering
a criterion more stringent than perceptibility.

Deflation involves subtracting $\sigma_ku_kv_k^\top$ from $X$.  This $u_kv_k^\top$
matrix is a noisy estimate of the $k$'th singular space of the signal matrix $S$.
A higher standard than perceptibility is to require that $u_kv_k^\top$
be an accurate estimate of the corresponding quantity of the matrix $S$.  
For this purpose we use a higher threshold for $\sigma_k$
that was obtained by \cite{perr:2009} and also by \cite{gavish2014optimal}.
This PGD threshold was derived for a white noise model.
In Section~\ref{pg} we extend it to a more general factor analysis
setup using ideas from the OptShrink algorithm \citep{nadakuditi2014optshrink}. We call the resulting method DDPA+.
We note that the Tracy-Widom scale factor $\tau_p$ of~\eqref{eq:imjtw} is too small to use in DDPA+. We provide the detailed algorithm in Algorithm \ref{ddpap}.


\begin{algorithm}

\caption{\textsc{DDPA+}: Deflated Deterministic Parallel Analysis+}
\begin{algorithmic}[1]
\State \textbf{input}: Data $X\in\real^{n\times p}$, centered, containing $p$ features of $n$ samples 
\State Initialize:  $k\gets0$. 
\State Compute singular values $\sigma_i(X)$, and eigenvalues $\lambda_i = \sigma_i^2$ sorted in decreasing order. Let  $\lambda = \sigma_1^2$, $r = \min(n,p)$.
\State Compute spectral functionals:

Stieltjes transforms: $$m = (r-1)^{-1} \sum_{i=2}^r (\lambda_i - \lambda)^{-1}, v = \gamma m -(1-\gamma)/\lambda,$$ 

D-transform: $D = \lambda m v$, 

Population spike estimate: $\ell = 1/D$, 

Derivatives of Stieltjes transforms: $$m' = (r-1)^{-1} \sum_{i=2}^r  (\lambda_i - \lambda)^{-2}, \,\, v' = \gamma m' +(1-\gamma)/\lambda^2$$ 

and derivative of D-transform $D' = mv+\lambda(mv'+m'v)$. 

\State Compute estimators of the squared cosines: $c_r^2 = m/(D'\ell)$, $c_l^2 = v/(D'\ell)$. 
\If{$\sigma_1^2 <4\ell^2 \cdot c_r^2 \cdot c_l^2$,}
\State $k\gets k+1$
\State $X \gets X - \sigma_1 u_1 v_1^\top$ (from the SVD of $X$) 
\State  Return to step 3. 
\EndIf
\State \textbf{return}: Selected number of factors $k$. 
\end{algorithmic}
\label{ddpap}
\end{algorithm}


\section{Numerical experiments}
\label{sec:numerical}

We perform numerical simulations to understand and compare the behavior of our proposed methods.  
We follow the simulation setup of \cite{dobriban2017factor}. 

\begin{figure}

\begin{subfigure}{.5\textwidth}
  \centering
  \includegraphics[scale=0.43]{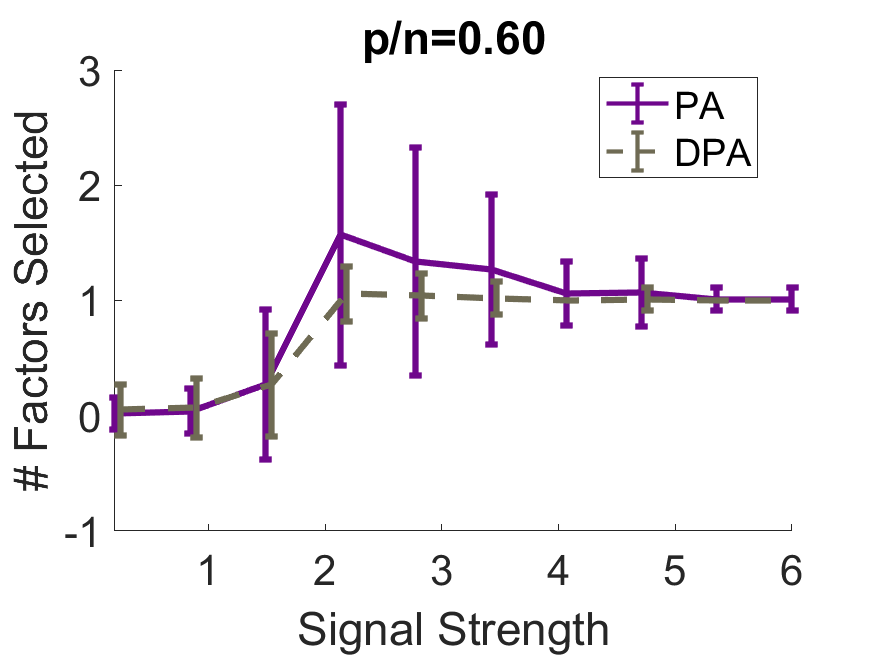}
  \caption{Effect of signal strength}
\end{subfigure}
\begin{subfigure}{.5\textwidth}
  \centering
  \includegraphics[scale=0.43]{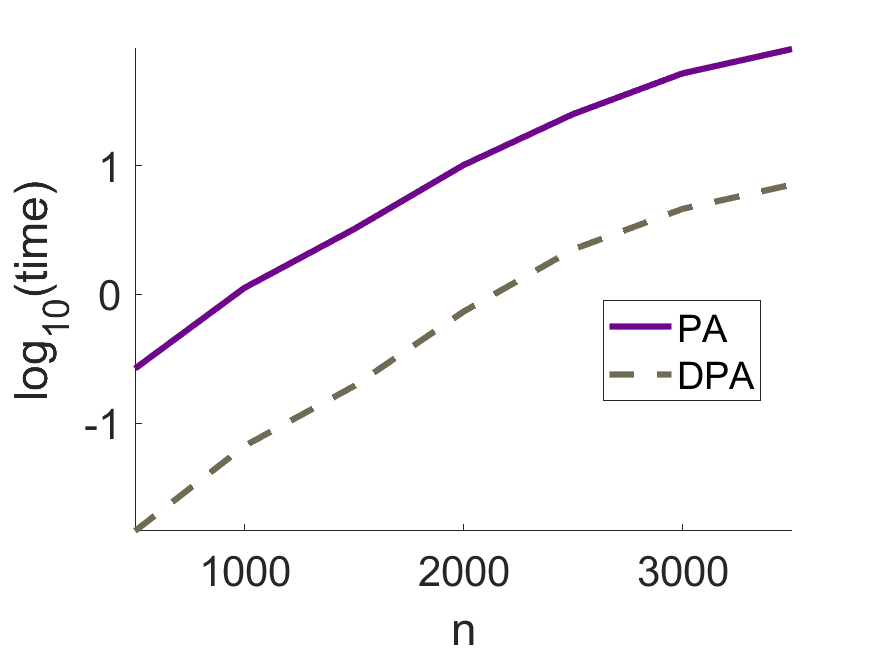}
  \caption{Running time}
\end{subfigure}
\caption{ (a)  Mean and $\pm1$SD of number of factors selected by PA and  DPA as a function of signal strength.
All SDs were zero. Small amount of jitter added for visual clarity.
 (b) Running time (in log-seconds) of  PA and  DPA.}
\label{fig:PA}
\end{figure}

\subsection{DPA versus PA}

First we compare DPA to PA.
For PA, we use the most classical version, generating 19 permutations, and selecting the $k$-th factor if $\sigma_k(X)$ is larger than the all of the permuted singular values. 

We simulate from the factor model $x_i = \Lambda \eta_i +\ep_i$. We generate
the noise $\ep_i \sim \N(0,T_p)$, where $T_p$ is a diagonal matrix of noise variances uniformly distributed on a grid between one and two.  The factor loadings are generated as $\Lambda = \theta \tilde Z$, 
where $\theta>0$ is a scalar corresponding to factor strength, and $\tilde Z$ is generated by normalizing
the columns of a random matrix
$Z \in\real^{p\times r}$ with IID $\N(0,1)$ entries. Also,  $\eta_i$ have IID standard normal entries. 

We use a one-factor model, so $r=1$, and work with sample size $n=500$ and dimension $p=300$. We report additional experiments with larger $p$ and non-Gaussian data in the Supplement. We take $\theta=\smash{\gamma^{1/2}\cdot s}$, for $s$ on a linear grid of $10$ points between $0.2$ and $6$ inclusive. This is the same simulation setup used in \cite{dobriban2017factor}.  We repeat the simulation $100$ times.


The results in Figure \ref{fig:PA}(a) show that PA and DPA  both
select the right number of factors as soon as the signal strength $s$ is larger than about $2$. This agrees with our theoretical predictions, since both
algorithms select the perceptible factors.  
The close match between PA and DPA confirms our view that PA
is estimating the upper edge of the spectrum.



A key advantage of DPA is its speed. PA and DPA both perform a first SVD of $X$ to compute the singular values $\sigma_k(X)$. This takes $O(np \min(n,p))$ flops. Afterwards, PA generates $n_{\perm}$ independent permutations $X_\pi$, and computes their SVD. This takes  $O(n_{\perm} \cdot np \min(n,p))$ flops. 
It would be reasonable to run both PA and DPA with a truncated
SVD algorithm computing only the first $\bar k$ singular values
where $\bar k$ is a problem specific a priori upper bound on $k$.
DPA would still have a cost advantage as that truncated SVD
would only need to be computed once versus $n_{\perm}+1$ times for PA.

After computing the SVD,
DPA only needs to compute $\upper(F_{\gamma_p, \hat  H_p})$, the  upper edge 
of the MP distribution. The fast method described in Section \ref{fast} has an approximate cost of $O(p)$ per iteration. We do not know whether the number of iterations required to achieve desired accuracy depends strongly on $p$.

In conclusion, the cost of DPA should be lower by a factor of the order of $n_{\perm}$ than that of PA. In empirical applications, the number of SVD computations is typically about 20, thus this represents a meaningful speedup. 

DPA is also much faster empirically. In Figure~\ref{fig:PA}(b), we report the results of a simulation where we increase the dimension $n$ and sample size $p$ keeping $p/n = \gamma = 0.6$, from $n=500$ to $n=3500$ in steps of 500. We use $n_{\perm} = 20$ permutations in FA. We only obtain one MC sample for each parameter setting.  We set the signal strength $\theta = 6\gamma^{1/2}$, so that the problem is ``easy'' statistically. Both methods are able to select the correct number of factors, regardless of the sample size (data not shown due to space limitations). Figure~\ref{fig:PA}(b) shows that DPA is much faster than PA. On a log-scale, we see that the improvement is between 10x and 15x. The expected improvement without the upper edge computation would be 20x, but the time needed for that computation reduces the improvement. However, DPA is still much faster than PA.

\subsection{DDPA versus DPA}

\begin{figure}
\begin{subfigure}{.5\textwidth}
  \centering
  \includegraphics[scale=0.43]{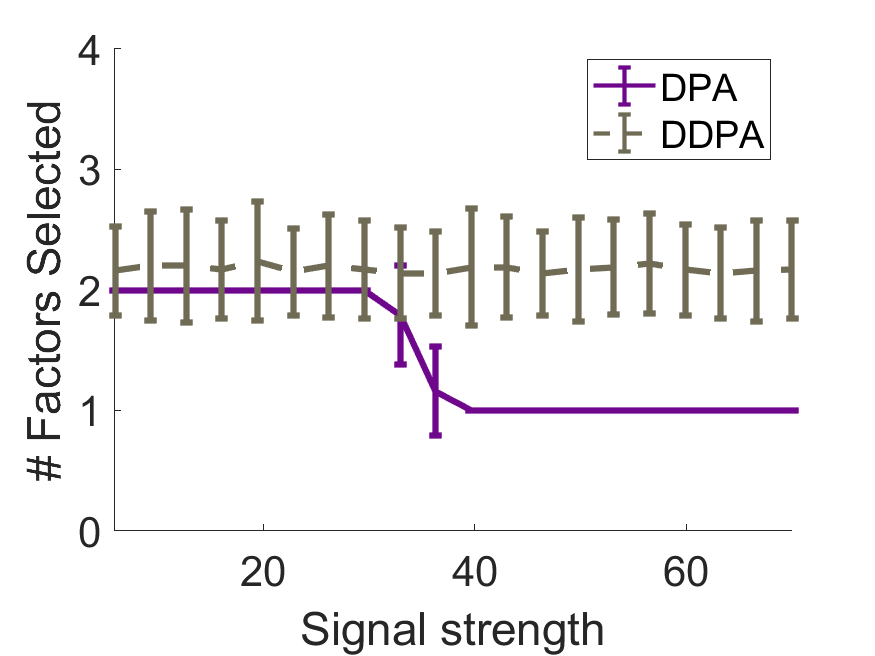}
  \caption{2-Factor model}
\end{subfigure}
\begin{subfigure}{.5\textwidth}
  \centering
  \includegraphics[scale=0.43]{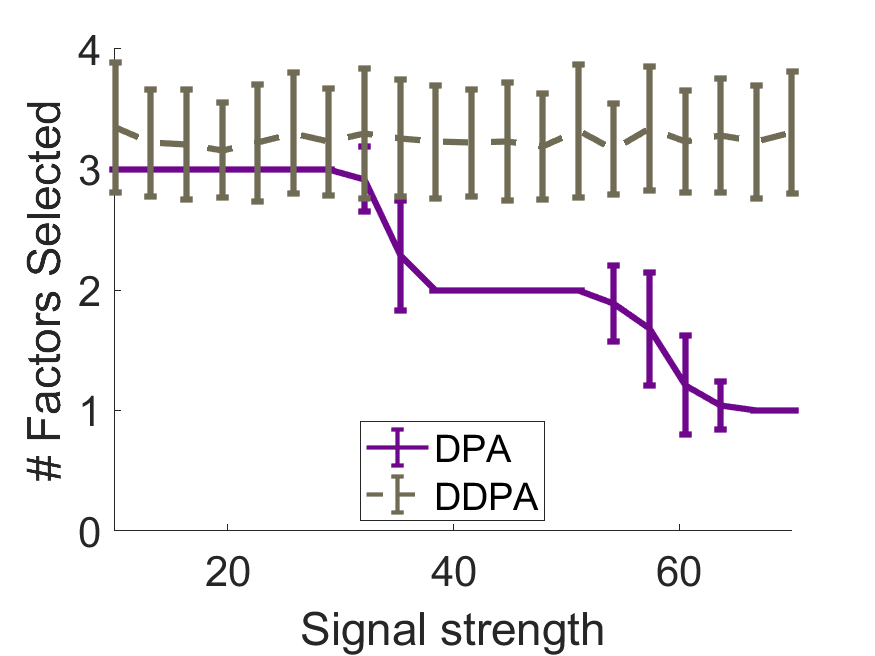}
  \caption{3-Factor model}
\end{subfigure}
\caption{Mean and $\pm 1$SD of the number of factors selected by DPA and DDPA as a function of the stronger factor value, with 2 and 3 factors.}
\label{fig:PA_shad}
\end{figure}

Here we show that DPA is affected by shadowing and that DDPA counters it.  We consider a two-factor and a three-factor model, with the same parameters as above, including heteroskedastic noise. For the two-factor model, we set the smaller signal strength to $\theta_1=6 \gamma^{1/2}$, which is a perceptible factor. We vary the larger factor strength $\theta_2=c_2 \gamma^{1/2}$ on a $20$ point grid from $c_2=6$ to $70$. This is the same simulation setup used in \cite{dobriban2017factor}.  The simulation was repeated $100$ times at each level. The results in Figure~\ref{fig:PA_shad}(a) show that DPA correctly finds both factors for small, but not for large, $\theta_2$. DPA begins to suffer from shadowing at $\theta_2\approx 30$ and shadowing is solidly in place by $\theta_2=40$.
In contrast, DDPA on average selects a constant number of factors, only slightly more than the true number, even when the larger factor is very strong. 
 
We see more scatter in the choices of $k$ by DDPA than by DPA.
Both DPA and DDPA are deterministic functions of $X$, but DDPA
subtracts a function of estimated singular vectors.
As we mentioned in Section~\ref{evec_estim}
those vectors are not always well estimated which could contribute
to the variance of DDPA. The SD for DPA is zero when all $100$ simulations picked the same $k$.

For the three-factor model, we set  $\theta_1=6 \gamma^{1/2}$, $\theta_2=10 \gamma^{1/2}$ and vary the largest factor $\theta_3=c_3 \gamma^{1/2}$ on the same $20$ point grid between $c_3=10$ and $70$. 
In Figure~\ref{fig:PA_shad}(b), we see the same behavior as for the two-factor model. 
For $\theta_3<30$ there is very little shadowing, by $\theta_3=40$ one factor
is being shadowed and by $\theta_3=50$ we start to see a second factor getting shadowed.
DDPA counters the shadowing, once again with more randomness than DPA has. In both cases $1$SD is approximately $0.5$ for DDPA.

\subsection{DDPA+ versus DDPA}

\begin{figure}
\begin{subfigure}{.5\textwidth}
  \centering
  \includegraphics[scale=0.43]{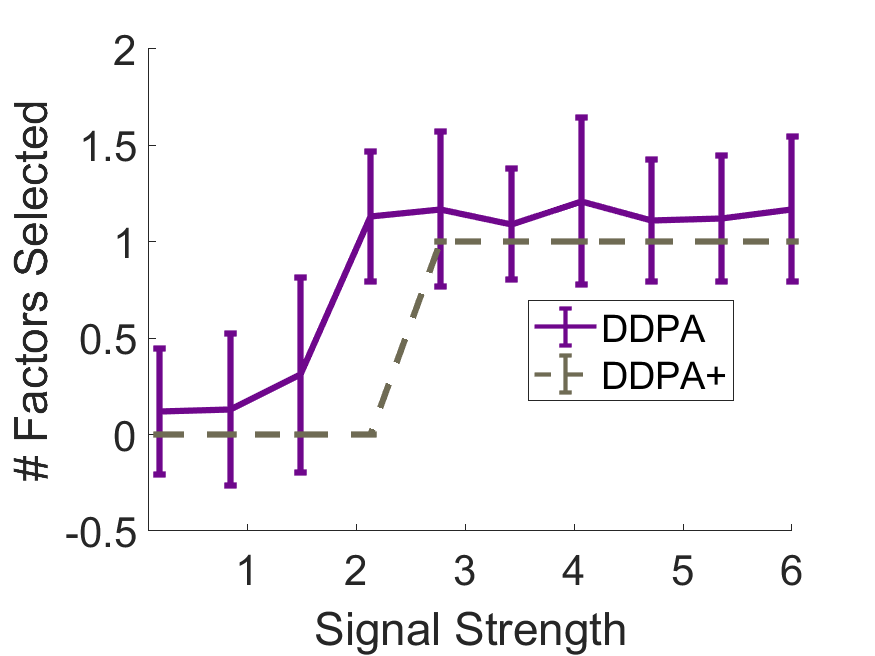}
  \caption{1-Factor model}
\end{subfigure}
\begin{subfigure}{.5\textwidth}
  \centering
  \includegraphics[scale=0.43]{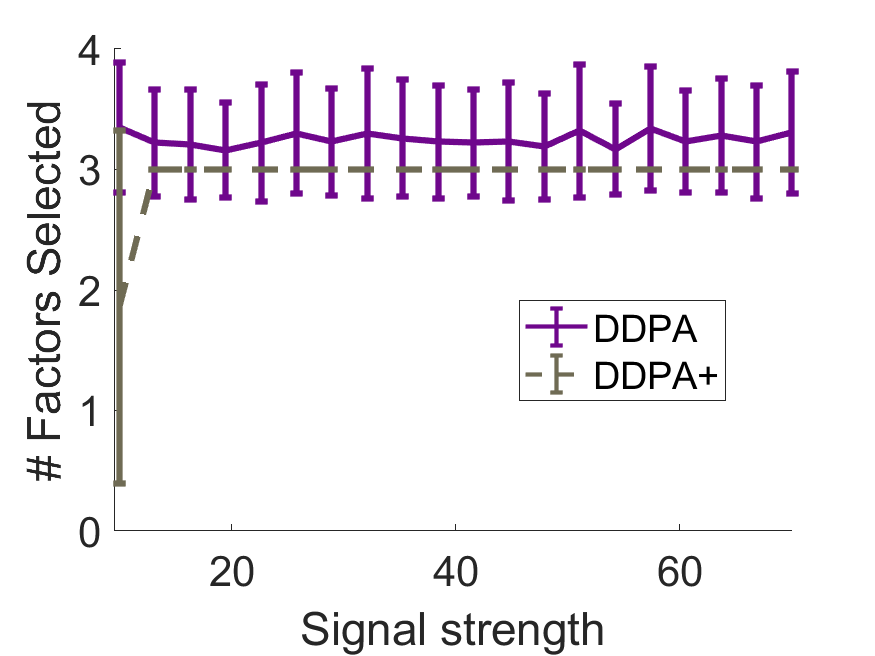}
  \caption{3-Factor model}
\end{subfigure}
\caption{Mean and $\pm 1$SD of the number of factors selected by DDPA as a function of stronger factor value, with 1 and 3 factors. With 
an increased threshold (DDPA+), and without (DDPA).}
\label{fig:PA_inc}
\end{figure}

Here we examine the behavior of DDPA+. We consider the same multifactor setting as in the previous section, now with a one-factor and a three-factor model. 

The results in Figure~\ref{fig:PA_inc}(a) show that the increased threshold
leads to some weak factors not being selected by DDPA+.  The threshold at
which a factor is included is increased compared to DDPA with
threshold, as expected.
For the three-factor model in Figure~\ref{fig:PA_inc}(b), increasing the threshold
counters the variability that we saw for DDPA in Figure~\ref{fig:PA_shad}(b).
In conclusion, increasing the threshold has the expected behavior.

Moreover, we see an increased variablity for DDPA+ in the leftmost setting on the right, when $\theta_3=\theta_2$. The problem is that the two population singular values are equal, and this is a problematic setting for choosing the number of factors. The problem disappears almost immediately
when $\theta_3$ is just barely larger than $\theta_2$ so then there is only a
small range of signal strengths where this happens.  We expect a very asymmetric
distribution of $\hat k$ there, so perhaps $\pm1$SD does not depict it well. This is a type of "shadowing from below".

\section{Data analysis}\label{sec:hgdp}
We consider the Human Genome Diversity Project (HGDP) dataset \cite[e.g.,][]{cann2002human,li2008worldwide}. The purpose of collecting this dataset was to evaluate the diversity in the patterns of genetic variation across the globe. We use the CEPH panel, in which SNP data was collected for 1043 samples representing 51 different populations from Africa, Europe, Asia, Oceania and the Americas. We obtained the data from \url{www.hagsc.org/hgdp/data/hgdp.zip}. We provide the data and processing pipeline on this paper's GitHub page. 

The data has $n=1043$ samples, and we focus on the $p=9730$ SNPs on chromosome 22. Thus we have an $n\times p$ data matrix $X$, where $X_{ij}\in\{0,1,2\}$ is the number of copies 
of the minor allele of SNP~$j$ in the genome of individual~$i$. We standardize the data SNP-wise, centering each SNP by its mean, and dividing by its standard error. For this step, we ignore missing values. Then,  we impute the missing values as zeroes, which are also equal to the mean of each SNP.
\begin{figure}
  \centering
  \includegraphics[scale=0.75]{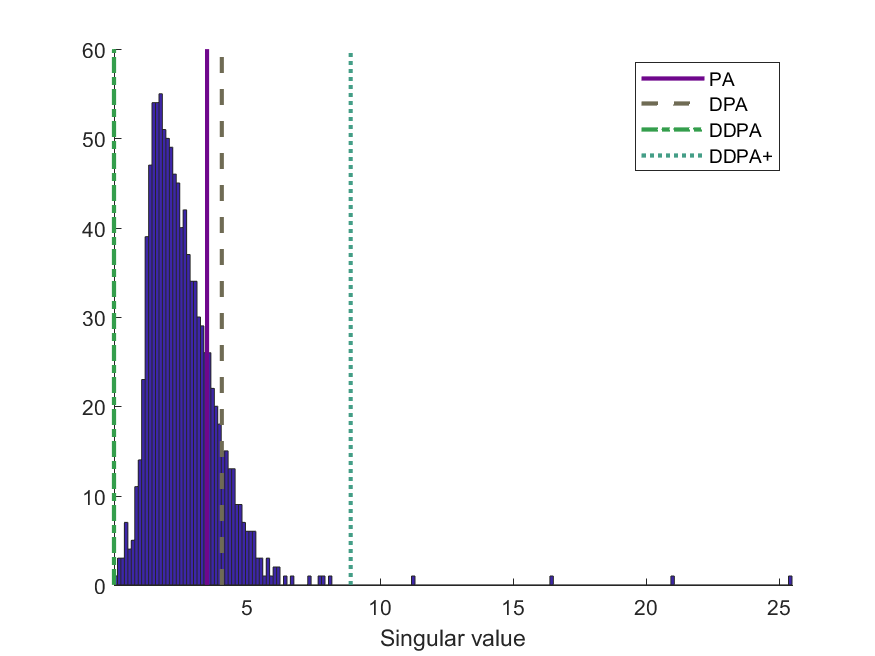}
\caption{Singular value histogram of the HGDP data. Overlaid  are the thresholds where selection of factors stops for the four methods.}
\label{fig:h}
\end{figure}

In Figure~\ref{fig:h}, we show the singular values of the HGDP data, and thresholds for selection of our methods: PA, DPA, DDPA and DDPA+. For PA, we use the sequential version, where the observed singular value must
exceed all $20$ permuted counterparts.
PA selects 212 factors, DPA 122, DDPA 1042, and DDPA+ 4. 


Because we have scaled the data, the generalized Marchenko-Pastur
distribution here reduces to the usual one.  The bulk edge of that
distribution is exceeded by all 122 factors selected by DPA.
PA selects even more factors.  The PA threshold is a noisy estimate
of the $1-212/9730\doteq 0.978$ quantile of the Marchenko-Pastur distribution.
Though their thresholds are not very far apart,  the density of sample singular values is 
high in that region, so PA and DPA pick quite different $k$.

Every time DDPA increases $k$ by one, the threshold it uses decreases and 
the cascade terminated with $k=n-1$ factors. This is the rank of our centered data matrix, and so the residual is zero.
There are 51 populations represented in the data and it is plausible that a very large
number of factors are present in the genome.  They cannot all be relatively important
and we certainly cannot estimate all those singular vectors accurately.
Restricting to well estimated singular vectors as DDPA+ does leads to $k=4$ factors.
The threshold is well separated from the one that would lead to $k=5$.


Since the ground truth is not known, we cannot be sure which $k$ is best.
We can however make a graphical examination.
Real factors commonly show some interesting structure, while we expect spurious factors
or poorly determined ones to appear Gaussian.
Of course, this is entirely heuristic. 
However, practitioners often attempt to interpret the PC estimators of the factors in exploratory FA, so this is a reasonable check. 
\cite{perry2010rotation} even base a test for factor structure on the non-Gaussianity
of the singular vectors, using random rotation matrices.

Figures \ref{fig:first} and \ref{fig:sec} show the left singular vectors of the data
corresponding to the top singular values. 
We see a clear clustering structure in at least the first $8$ PCs. 
DDPA+ selects only $4$ factors which we interpret as being conservative about
estimating the singular vectors beyond the $4$'th.
While it is visually apparent that there is some structure in eigenvectors
$5$ through $8$, DDPA+ stopped because that structure could not
be confidently estimated.
There is much less structure in the PCs beyond the $8$'th.
Thus, we believe that PA and DPA select many more factors than can be well
estimated from this data even if that many factors exist.
DDPA failed completely on this data, unlike in our simulations, where it was merely variable.
Here, the histogram of the bottom singular values never
looked like a generalized Marchenko-Pastur. We think DDPA+ is more robust.

To be clear, we think that the main reason why DDPA does not work is that there is a lot of clustering in the data. There are samples from several different populations from
Europe, Asia, and Africa. Within each continent, there are samples from different regions, and within those there are samples from different countries.  The IID assumption is not accurate. Instead, it would be more accurate to use a hierarchical mixture model. Simulations in the appendix show how DDPA fails under such a model. However, developing methods for such models is beyond our current scope. 

%
%

\newcommand{\FW}{0.47\textwidth}
\newcommand{\TRA}{0}
\newcommand{\TRB}{0}
\newcommand{\TRC}{0}
\newcommand{\TRD}{0}
\newcommand{\PBW}{0.3}

\begin{figure}[p]
\centering
\begin{tabular}{ccc}
{\begin{sideways}\parbox{\PBW\columnwidth}{\centering $1$--$4$}\end{sideways}} &
\includegraphics[width=\FW, trim = \TRA mm \TRB mm \TRC mm \TRD mm, clip = TRUE]{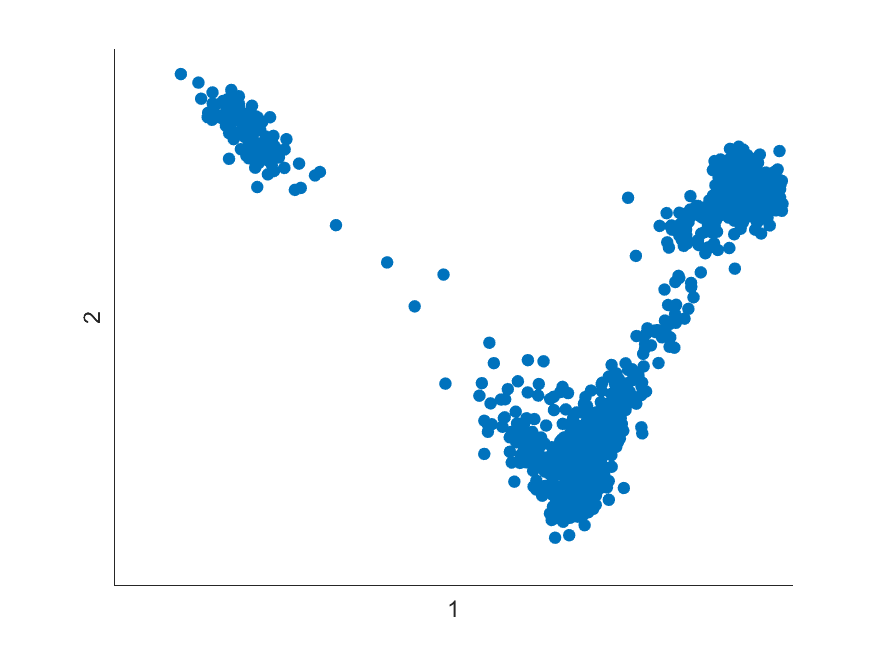} &
\includegraphics[width=\FW, trim = \TRA mm \TRB mm \TRC mm \TRD mm, clip = TRUE]{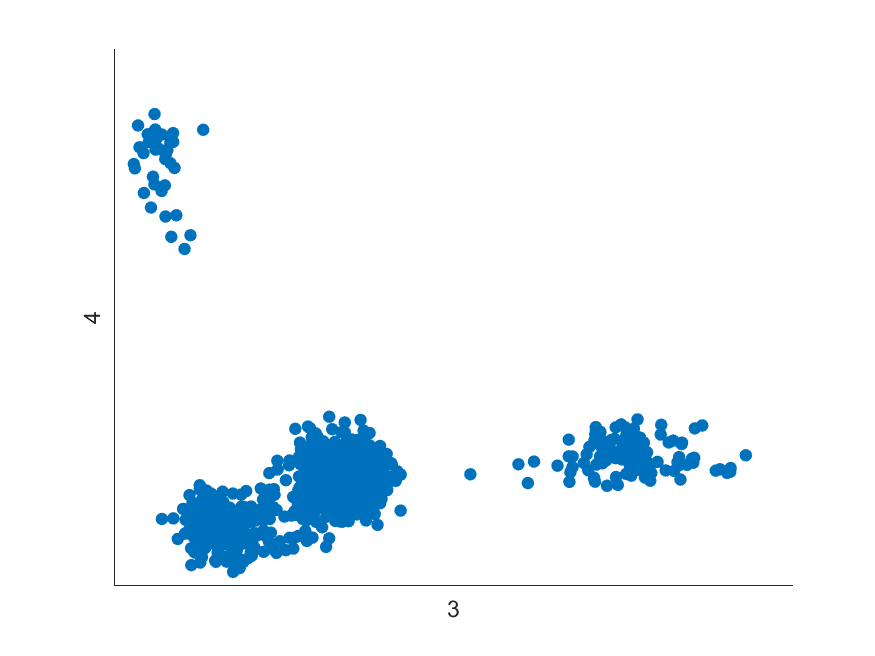} \\
{\begin{sideways}\parbox{\PBW\columnwidth}{\centering $5$--$8$}\end{sideways}} &
\includegraphics[width=\FW, trim = \TRA mm \TRB mm \TRC mm \TRD mm, clip = TRUE]{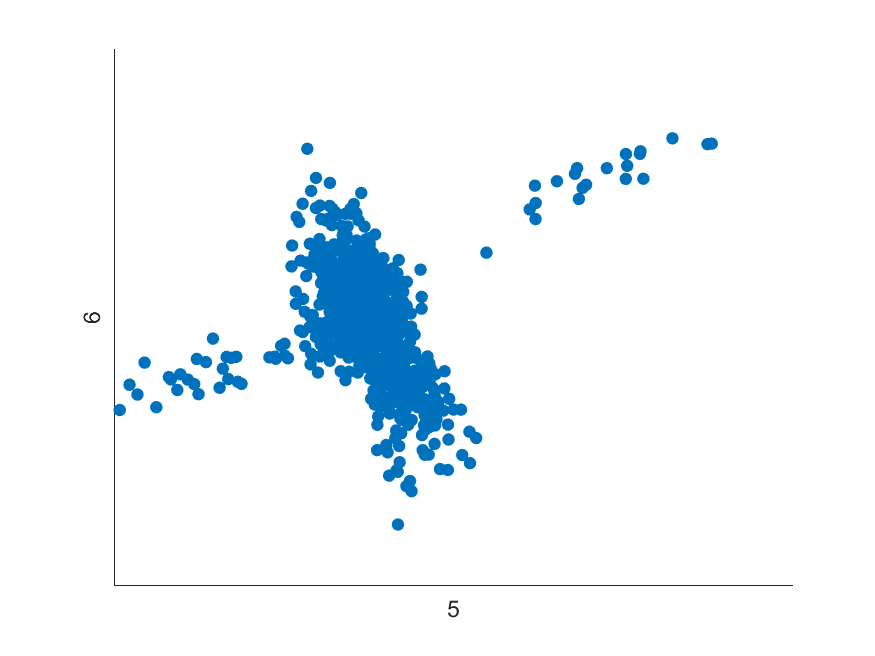} &
\includegraphics[width=\FW, trim = \TRA mm \TRB mm \TRC mm \TRD mm, clip = TRUE]{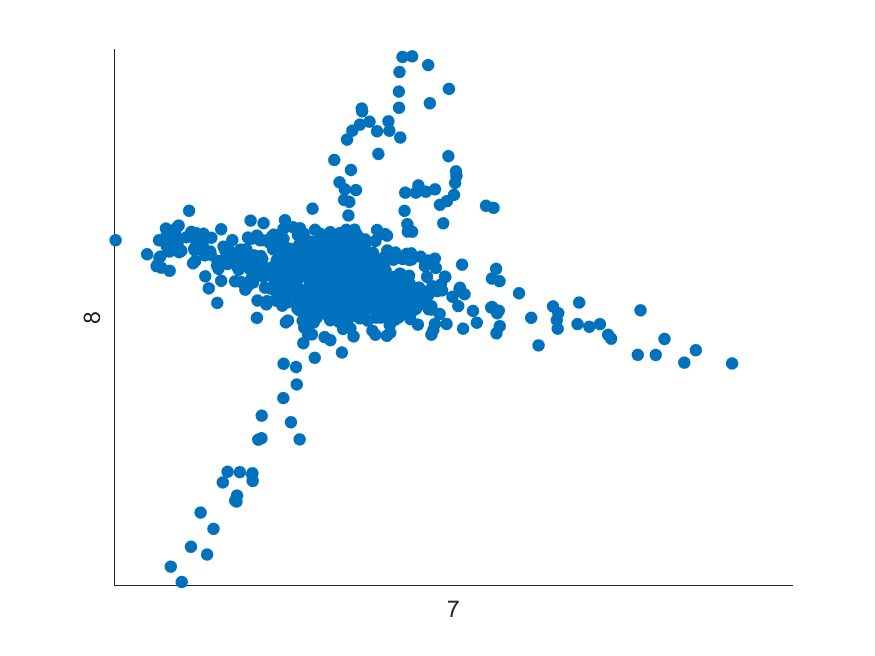}  \\
{\begin{sideways}\parbox{\PBW\columnwidth}{\centering $9$--$12$}\end{sideways}} &
\includegraphics[width=\FW, trim = \TRA mm \TRB mm \TRC mm \TRD mm, clip = TRUE]{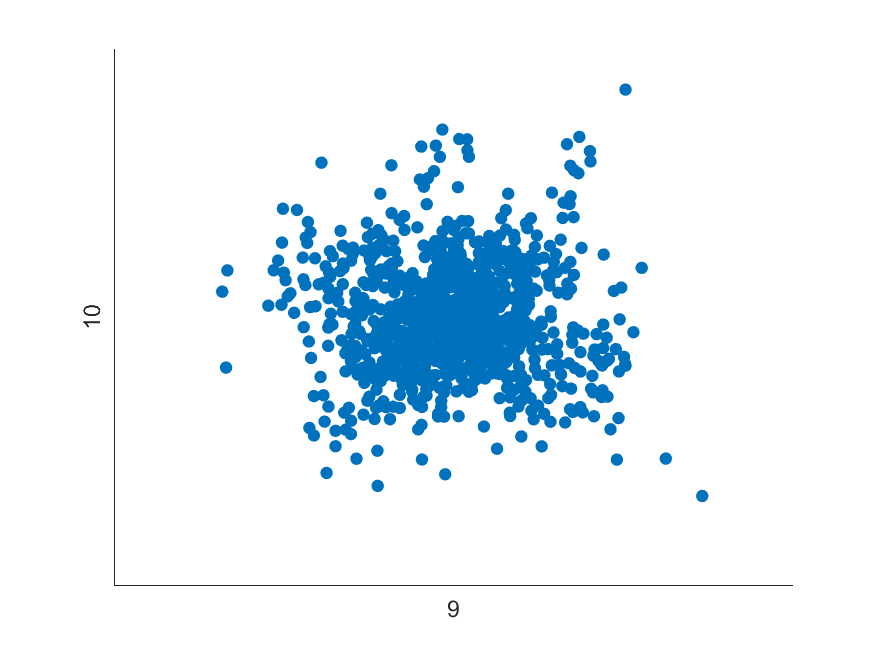} &
\includegraphics[width=\FW, trim = \TRA mm \TRB mm \TRC mm \TRD mm, clip = TRUE]{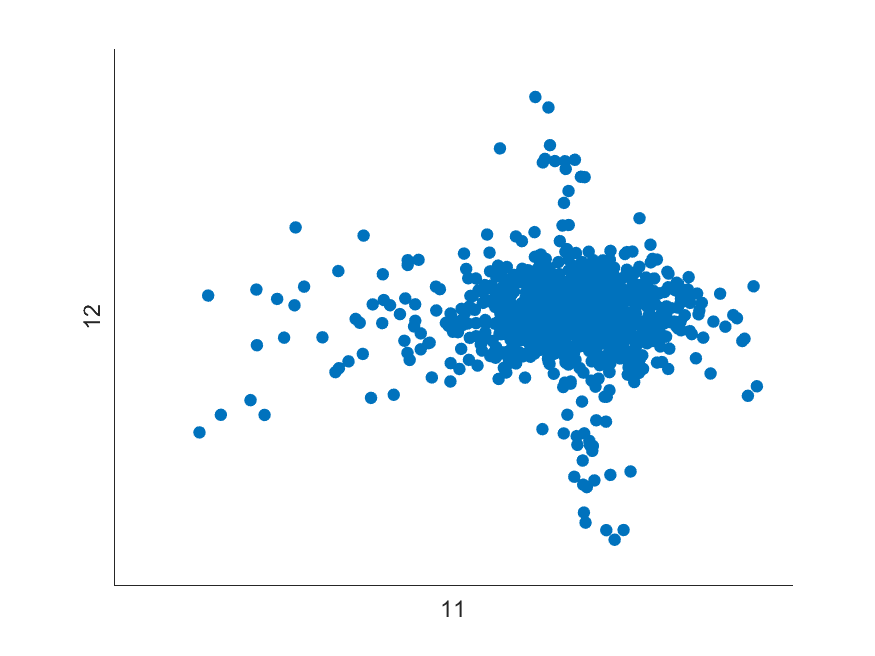} \\
\end{tabular}
\caption{Scatterplots of the left singular vectors 1--12 of the HGDP data. Each point represents a sample. }
\label{fig:first}
\end{figure}

\begin{figure}[p]
\centering
\begin{tabular}{ccc}
{\begin{sideways}\parbox{\PBW\columnwidth}{\centering $13$--$16$}\end{sideways}} &
\includegraphics[width=\FW, trim = \TRA mm \TRB mm \TRC mm \TRD mm, clip = TRUE]{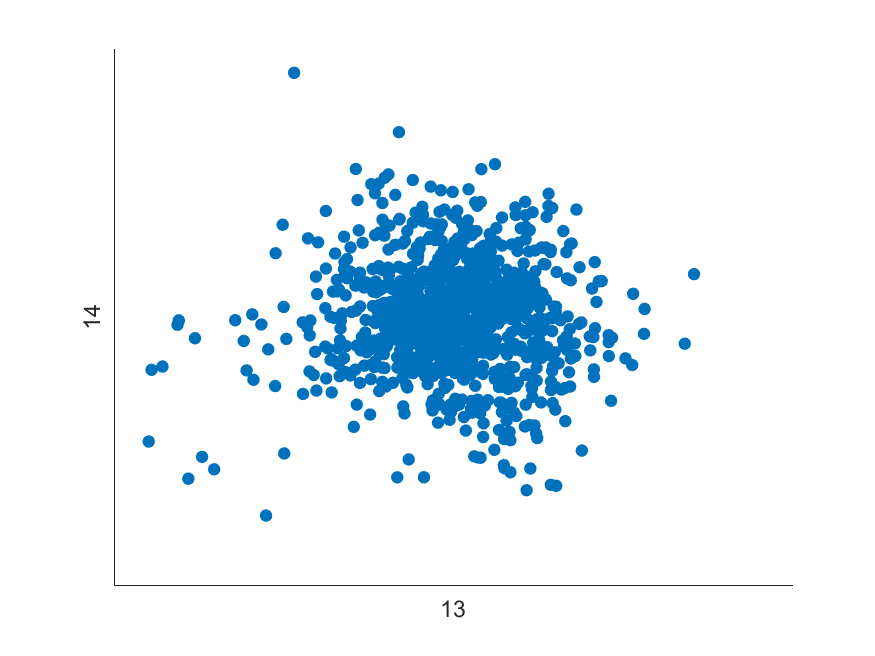} &
\includegraphics[width=\FW, trim = \TRA mm \TRB mm \TRC mm \TRD mm, clip = TRUE]{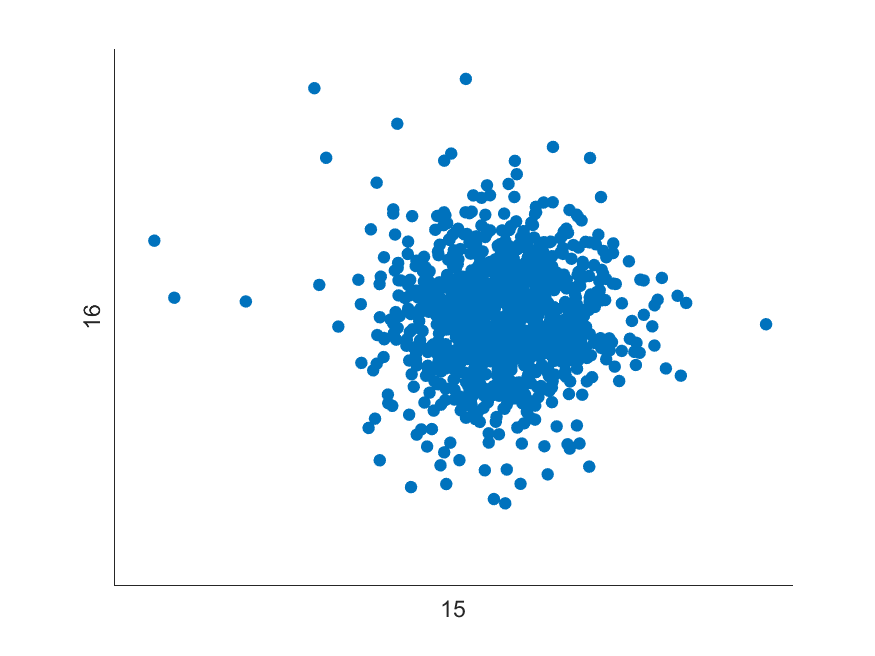} \\
{\begin{sideways}\parbox{\PBW\columnwidth}{\centering $17$--$20$}\end{sideways}} &
\includegraphics[width=\FW, trim = \TRA mm \TRB mm \TRC mm \TRD mm, clip = TRUE]{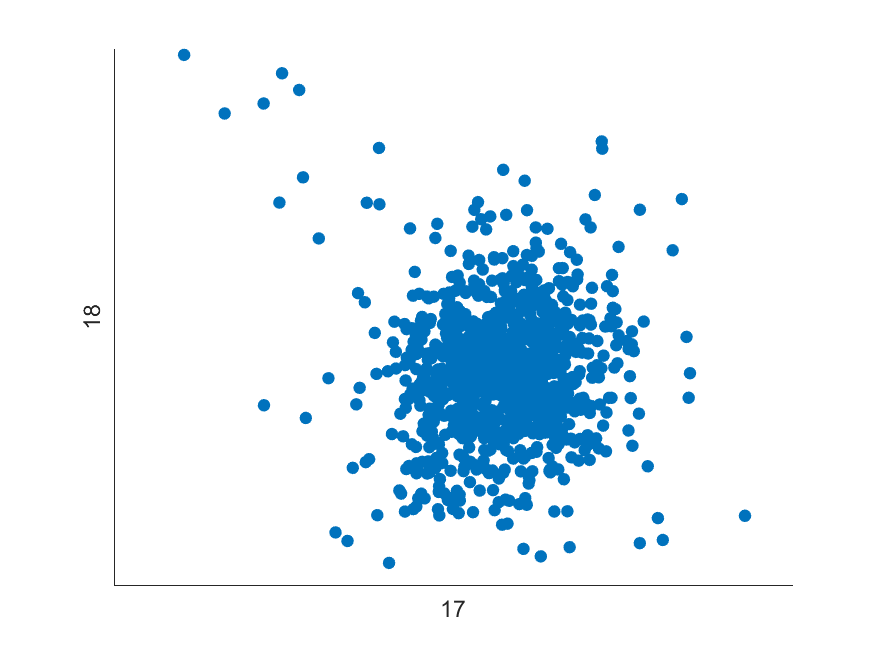} &
\includegraphics[width=\FW, trim = \TRA mm \TRB mm \TRC mm \TRD mm, clip = TRUE]{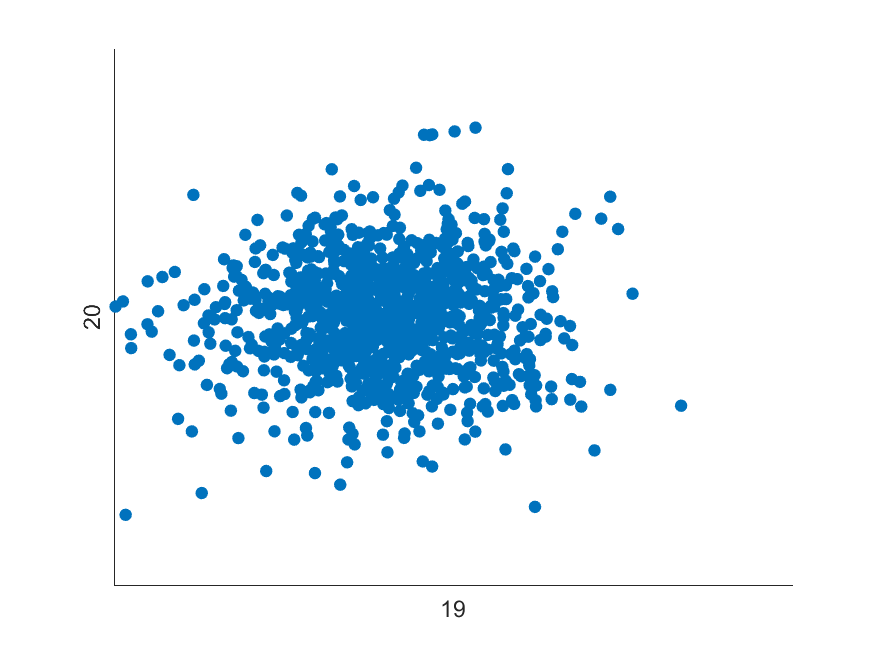}  \\
{\begin{sideways}\parbox{\PBW\columnwidth}{\centering $21$--$24$}\end{sideways}} &
\includegraphics[width=\FW, trim = \TRA mm \TRB mm \TRC mm \TRD mm, clip = TRUE]{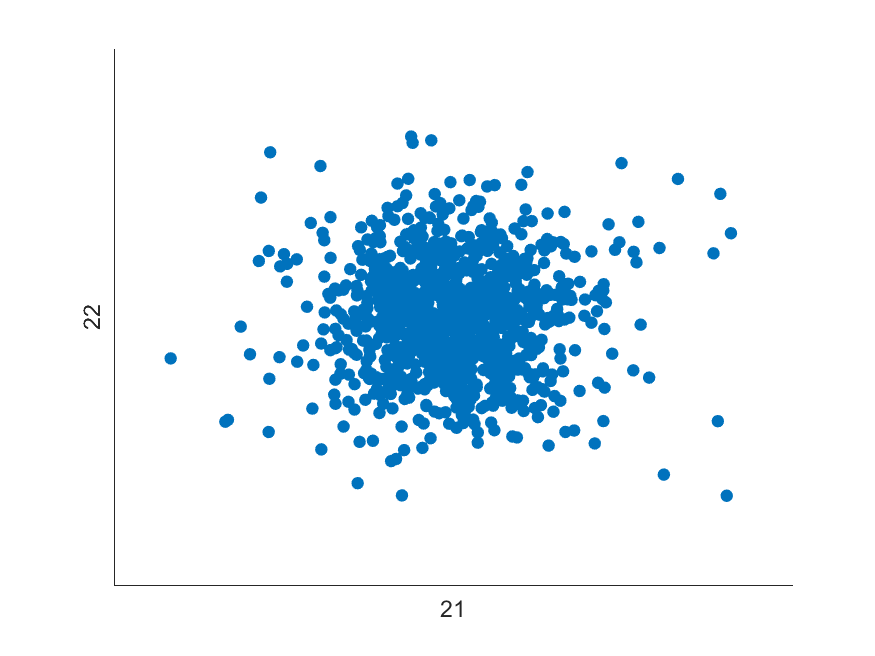} &
\includegraphics[width=\FW, trim = \TRA mm \TRB mm \TRC mm \TRD mm, clip = TRUE]{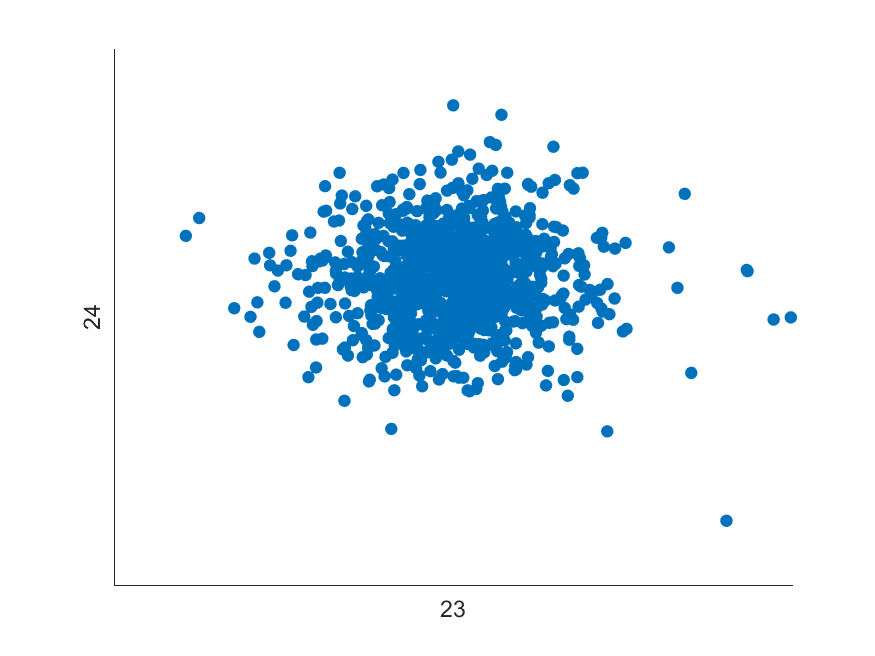} \\
\end{tabular}
\caption{Scatterplots of the left singular vectors 13--24 of the HGDP data.}
\label{fig:sec}
\end{figure}

\section{Computing the thresholds}\label{sec:getthresh}

In this section we describe how our thresholds are computed.
First we consider a fast method to compute the DPA threshold.
Then we describe how to compute the DDPA+ threshold.

\subsection{Fast computation of the upper edge}\label{sec:fast_edge}
\label{fast}
We describe a fast method to compute the upper edge $\upper(F_{\gamma, H})$ of the Marchenko-Pastur distribution, given the population spectrum $H$ and the aspect ratio $\gamma$. Our method described here is, implicitly, one of the steps of the Spectrode method, which computes the density of the whole MP distribution $F_{\gamma, H}$  \citep{dobriban2015efficient}. However, this subproblem was not considered separately in \cite{dobriban2015efficient}, and thus calling the Spectrode method can be inefficient when we only need the upper edge. 

We are given a probability distribution $H = \sum_{j=1}^p w_j \delta_{\Phi_j}$ which is a mixture of point masses at $\Phi_j\ge 0$, with $w_j\ge0$ and $\sum_{j=1}^p w_j = 1$. We are also given the aspect ratio $\gamma>0$. As a consequence of the results in \cite{silverstein1995analysis}, see e.g., \citet[Lemma 6.2]{bai2009spectral}, the upper edge $\upper(F_{\gamma, H})$ has the following characterization. Let $B = -1/\max_j \Phi_j$, and assume that $H \neq \delta_0$, so that $B<0$ is well-defined. Consider the map 
$$z(v) = -\frac{1}{v}  + \gamma\sum_{j=1}^{p} \frac{w_j\Phi_j}{1 + \Phi_jv}.$$
On the interval $(B,0)$, $z$ has a unique minimum $v^*$. Then, $\upper(F_{\gamma, H}) = z(v^*)$. Moreover, $z$ is strictly convex on that interval, and asymptotes to $+\infty$ at both endpoints. This enables finding $v^*$ numerically by solving $z'(v)=0$ using any one-dimensional root finding algorithm, such as bisection 
or Brent's method \citep{bren:1973}.  

\subsection{The DDPA+ threshold}\label{pg}


For simplicity we consider discerning whether our $X$ from which $k-1$ factors have been subtracted
contains an additional $k$'th well estimated factor or not.
Let $n^{-1/2}X = \theta ab^\top +N$, where $S = \theta ab^\top$ is the signal, with unit vectors $a$, $b$, and $N$ is the noise. Let $X = \sum_{i=1}^r\sigma_i u_i v_i^\top$, where $r=\min(n,p)$, be the SVD of~$X$. 

We compare two estimators of $S$: the first empirical singular matrix $\hat S = \sigma_1 u_1 v_1^\top$, and the $n\times p$ zero matrix $0_{n\times p}$, which corresponds to a hard thresholding estimator on the first singular value. 
Now
$\hat S$ is  more accurate than $0$ in mean squared error, when
$$\Vert\sigma_1u_1v_1^\top - \theta a b^\top\Vert_{F}^2< \Vert\theta a b^\top\Vert_{F}^2.$$
By expanding the square, isolating $\sigma_1$, and squaring, the criterion becomes
\begin{align}\label{eq:criterion}
\sigma_1^2 <4\theta^2 \cdot (u_1^\top a)^2\cdot (v_1^\top b)^2.
\end{align}

The theory of spiked covariance models \citep[e.g.,][]{benaych2012singular}, 
shows that as $n,p \to \infty$, under the conditions of the theorems in our paper and for fixed factor strengths, the empirical spike singular value $\sigma_1$ and the empirical squared cosines $(u_1^\top a)^2,(v_1^\top b)^2$ have a.s. limits. Moreover, the OptShrink method \citep{nadakuditi2014optshrink} provides consistent estimators of the unknown population spike $\theta^2$, and the squared cosines $(u_1^\top a)^2,(v_1^\top b)^2$. OptShrink relies on the empirical singular values $\sigma_i$, $i=1,\ldots,r$. To estimate the unknown population spike $\theta^2$, it uses the singular values $\sigma_i$, $i=m,\ldots,r$, for some prespecified $m$. Here we take $m=2$. 

The specific formulas can be found in \cite{benaych2012singular}. We provide them here for completeness
and  they are also in the code at \url{github.com/dobriban/DPA}. The full algorithm is stated separately in the paper. 
Let $\lambda = \sigma_1^2$, and $\lambda_i = \sigma_i^2$, $i=2\ldots,r$, where $r = \min(n,p)$. Then, let $m = (r-1)^{-1} \sum_i (\lambda_i - \lambda)^{-1}$, $v = \gamma m -(1-\gamma)/\lambda$, $D = \lambda m v$, $\ell = 1/D$, $m' = (r-1)^{-1} \sum_i (\lambda_i - \lambda)^{-2}$, $v' = \gamma m' +(1-\gamma)/\lambda^2$, and $D' = mv+\lambda(mv'+m'v)$. The estimators of the squared cosines between the true and empirical left and right singular vectors are, respectively, $c_r^2 = m/(D'\ell)$, $c_l^2 = v/(D'\ell)$.

\section{Discussion}\label{sec:disc}

In this paper we have developed a deterministic counterpart, DPA, to the popular PA
method for selecting the number of factors.  DPA is completely deterministic
and hence reproducible. It is also much faster than PA.  Both of these traits make it
a suitable replacement for PA when factor analysis is used as one step in a lengthy
data analysis workflow.
Under standard factor analysis assumptions, DPA inherits PA's vanishing probability
of selecting false positives proved in \cite{dobriban2017factor}.  

PA and DPA have a shadowing problem that we can remove by deflation. The resulting DDPA
algorithm counters deflation but chooses a more random number of factors because
it subtracts poorly estimated singular vectors. We counter this in our DDPA+ algorithm
by raising the threshold to only include factors whose singular vectors are well determined.


Factor analysis is used in data analysis pipelines for multiple purposes, including factor interpretation,
data denoising, and data decorrelation.  The best choice of $k$ is probably different for each of
these uses.  Denoising 
has been well studied but we think more work is needed for the other goals.

\section*{Acknowledgments}

This work was supported by the U.S.\ NSF under grants
DMS-1521145 and DMS-1407397. The authors are grateful to Jeha Yang for pointing out an error in an earlier version of the manuscript.

\section{Proof outlines}\label{sec:supp}

In this section we provide some proof outlines. The remaining proofs are in the Appendix in the Supplementary Material. 

\subsection{Proof of Theorem \ref{thm:fa_cons}}
\label{fa_cons_pf}

\subsubsection{General theory}

Recall that $n^{-1/2}X = S+N$ for signal and noise matrices $S$ and $N$, where $\Vert N\Vert_2\to b>0$ a.s.
We begin with a technical lemma that provides a simple condition for DPA to select the perceptible factors.

\begin{lemma}[DPA selects the perceptible factors] 
If $\upper(F_{\gamma_p, \hat H_p})^{1/2} \to b\in(0,\infty)$ almost surely,
then DPA selects all perceptible factors
and no imperceptible factors, almost surely. 
\label{lem:DPA_sig}
\end{lemma}

%
The proof is immediate.  In the remainder, we will provide concrete conditions when $\upper(F_{\gamma_p, \hat  H_p})^{1/2} \to b$ holds. Recall that the noise has the form $N = n^{-1/2} Z \Phi^{1/2}$, where the entries of $Z$ are independent random variables of mean zero and variance one.  Then the true (population) variances are the entries of the diagonal matrix $\Phi$, and we let $H_p =\esd(\Phi)$. 

\begin{prop}[Noise operator norm]
\label{prop:noise_op_norm}
Let $N = n^{-1/2} Z \Phi^{1/2}$,  where the entries of $Z$ are independent standardized random variables with bounded $6+\eps$-th moment for some $\eps>0$, and 
$\Phi\in\real^{p\times p}$ is a diagonal positive semi-definite matrix with $\max_j\Phi_{j}\le C<\infty$.
Let $p\to\infty$ with $p/n \to \gamma >0$, while $\esd(\Phi)\to H$
and $\max_i\Phi_i\to \upper(H)$ where the limiting distribution
has $\upper(H)<\infty$.
Then $\Vert N\Vert_2 \to \upper(F_{H,\gamma})^{1/2}$ almost surely.
\end{prop}
\begin{proof}
This essentially follows from Corollary 6.6 of \cite{bai2009spectral}. 
A small modification is needed to deal with the non-IID-ness, as explained in \cite{dobriban2017optimal}.
\end{proof}

Using Proposition~\ref{prop:noise_op_norm}
we can ensure that $\upper(F_{\gamma_p, \hat  H_p})^{1/2} \to b$, as required in
Lemma~\ref{lem:DPA_sig}, by showing that $\upper(F_{\gamma_p, \hat  H_p}) \to \upper(F_{\gamma, H})$. We now turn to this. 

\subsubsection{Upper edge of MP law}
In this section, we will provide conditions under which
$\upper(F_{\gamma_p, \hat  H_p}) \to \upper(F_{\gamma, H})$. This statement requires a delicate new analysis of the MP law, going back to the first principles of the Marchenko-Pastur-Silverstein equation. 

\begin{theorem}[Convergence of upper edge] 
Suppose that as $n,p\to\infty$,
$\gamma_p=p/n\to\gamma\in(0,\infty)$,
$\hat  H_p \tod H$, and
$\upper(\hat  H_p) \to \upper(H)<\infty$.
Then  $\upper(F_{\gamma_p, \hat  H_p}) \to \upper(F_{\gamma, H})$.
\label{thm:upper_edge_conv}
\end{theorem}
\begin{proof}
See the supplement. 
\end{proof}

Theorem \ref{thm:upper_edge_conv} provides sufficient conditions guaranteeing the convergence of the MP upper edge. For normalized data 
$n^{-1/2}X = S+N$ with $N=n^{-1/2} Z\Phi^{1/2}$, as in the setting of Proposition~\ref{prop:noise_op_norm}, these conditions are satisfied provided that the signal columns of $S$ have vanishing norms.

\begin{prop}\label{prop:ecdf_conv}
Let the scaled data be $n^{-1/2}X = S+N$ 
with noise $N=n^{-1/2} Z\Phi^{1/2}$.
Let $\Phi$ and $Z$ obey the conditions of Proposition~\ref{prop:noise_op_norm}
except that the moment condition on $Z$ is increased from $6+\eps$
to $8+\eps$ bounded moments for some $\eps>0$.
Let $S$ have columns $s_1,\dots,s_p$, and 
suppose that $\max_j \Vert s_j\Vert_2 \to 0$. Then $\hat  H_p \tod H$ and $\upper(\hat  H_p) \to \upper(H)$.
\end{prop}
\begin{proof}
See the supplement. 
\end{proof}

Theorem \ref{thm:upper_edge_conv} and Proposition \ref{prop:ecdf_conv} provide a concrete set of assumptions that can be used to prove our main result for factor models. 

\subsubsection{Finishing the proof of Theorem \ref{thm:fa_cons}}

It remains to check that the conditions of Propositions 
\ref{prop:noise_op_norm} and \ref{prop:ecdf_conv}
hold. In matrix form, the factor model reads $X = U \Psi^{1/2}\Lambda^\top +Z\Phi^{1/2}$. We first normalize it to have operator norm of unit order: $n^{-1/2}X = n^{-1/2}U \Psi^{1/2}\Lambda^\top +n^{-1/2}Z\Phi^{1/2}$. 

We need to verify conditions on signal and noise.
The assumption about the idiosyncratic terms in Theorem~\ref{thm:fa_cons}
satisfies the conditions of  Proposition \ref{prop:ecdf_conv} which
are stricter than those of  Proposition \ref{prop:noise_op_norm}.
So the noise conditions are satisfied by assumption. 

Now we turn to the signal $S$.  It has columns $s_1,\dots,s_p$. For Proposition \ref{prop:noise_op_norm}, we need to show
that $M=\max_j\Vert s_j\Vert\to0$.
Let $\Lambda\Psi^{1/2} = [d_1,\ldots,d_r]$.  
Then
$$S=n^{-1/2}  U \Psi^{1/2}\Lambda^\top
= n^{-1/2} \sum_{\ell=1}^r u_\ell d_\ell^\top.$$ 
Now $s_j = n^{-1/2} \sum_{\ell=1}^r u_\ell d_\ell(j)$, where $d_\ell(j)$ denotes the $j$-th entry of $d_\ell$. 
Thus 
$$\Vert s_j\Vert \le n^{-1/2}   \sum_{\ell=1}^r \Vert u_\ell\Vert \cdot|d_\ell(j)|
\le n^{-1/2}  \max_{1\le\ell\le r} \Vert u_\ell\Vert \cdot \max_{1\le j\le p}   \sum_{\ell=1}^r |d_\ell(j)|.$$  
Next, by definition $\max_j     \sum_{\ell=1}^r |d_\ell(j)| = \Vert\Lambda\Psi^{1/2}\Vert_\infty$ which approaches $0$
by condition 5 on factor loadings in Theorem~\ref{thm:fa_cons}.
It now suffices to show that
$\max_\ell \Vert u_\ell\Vert = O(n^{1/2})$.

Each $u_\ell$ has IID entries with mean 0 and variance 1, so 
$n^{-1/2} \Vert u_\ell\Vert \to 1$ a.s.\ by the LLN. 
The problem is to guarantee that \emph{the maximum} of $r$ such random variables also converges to 1. If $r$ is fixed, then this is clear. A more careful analysis allows  $r$ to grow subject to the limits in Assumption 4 of the theorem,
and using the bounded moments of order $\nu=4+\delta$ for the
entries of $U$ given by assumption 1.
Specifically, with 
$M_u(p)=M_u = \max_{1\le\ell\le r} n^{-1/2} \Vert u_\ell\Vert$ 
we can write
\begin{align*}
\Pr(M_u > 2) &\le r \max_{1\le \ell\le r}\Pr\bigl( \Vert u_\ell\Vert^2-n > 3n\bigr) \\
&\le r  \max_{1\le\ell\le r}\E\bigl((\Vert u_\ell\Vert^2-n)^{\nu/2}\bigr)/ (3n)^{\nu/2} \\
&\le C rn^{-\nu/2},
\end{align*}
for some $C<\infty$.
Thus $M_u(p)$ is a.s.\ bounded, if the sequence $rn^{-\nu/2}$ is summable.
Summability holds by the assumption that the sequence $r/n^{1+\delta/4}$ is summable.
This finishes the proof. 



\section{Appendix}\label{sec:app}

In this section, we prove the remaining theorems and some supporting
results.  Familiarity with random matrix theory 
is assumed. See for instance \cite{yao2015large}. We also provide some additonal numerical results, and algorithm details.

By  $\cplx^+$ we mean the set
$\{ z\in \cplx\mid \im(z)>0\}$.

\subsection{Proof of Theorem \ref{thm:upper_edge_conv}}
\label{sec:upper_edge_conv_pf}
The Stieltjes transform of $F_{\gamma, H}$  is
\begin{align}\label{eq:stieltjes}
m_{\gamma,H}(z) = \E_{F_{\gamma, H}}((X-z)^{-1})
\end{align} 
and we let $v_{\gamma,H}(z) = \gamma m_{\gamma,H}(z)-(1-\gamma)/z$ be its companion Stieltjes transform.  
The Silverstein equation 
\citep{marchenko1967distribution,silverstein1995analysis} for $H$ shows that  for 
$z\in\cplx^+$, $v_{\gamma,H}$ is the unique solution in $\cplx^+$ to
\begin{align}\label{eq:silv}
z= -\frac{1}{v_{\gamma,H}} + \gamma \int \frac{t \mrd H(t)}{1 + tv_{\gamma,H}}.
\end{align}
Let $G = G(v_{\gamma,H})$ be the function defined on the right hand side
of~\eqref{eq:silv}. Then $G$ is the functional inverse of the map $z \to v_{\gamma,H}(z)$ on its domain of definition.

From \citet[Chapter 6]{bai2009spectral}  it follows that the supremum of the support of $F_{\gamma, H}$ can be expressed in a simple form in the following way. Consider intervals $I_t=(t,0)$ for $t<0$ such that $G'(x)<0$ for all $x\in I_t$. It is easy to see that the set of such intervals is nonempty. Let $t^*$ be the infimum of such $t$. Then,
$$\upper(F_{\gamma, H}) = G(t^*).$$

From \citet[Chapter 6]{bai2009spectral}, 
it also follows that the structure of $G$ is such that $G$ has a singularity at 0, and one or more singularities at a discrete set of values $t_i <0$ that does not have a largest accumulation point. Among these, let $t_1$ be the largest value. Then, $G(x)>0$, for all $x\in(t_1,0)$, $G(x) \to +\infty$ as $x\to 0$ or $x \to t_1$. Moreover, $t_1<t^*<0$, and $G'(x)<0$ for $x\in(t_1,t^*)$, $G'(x)>0$ for $x\in(t^*,0)$. 

Based on the above structure, in order to show the convergence $\upper(F_{\gamma_p, \hat  H_p}) \to \upper(F_{\gamma, H})$, it is enough to show that $G_{\gamma_p, \hat  H_p}$ converges to $G_{\gamma,H}$ uniformly on any compact sub-interval $[l,u]\subset (t_1,0)$.  

To show uniform convergence, we can write 
\begin{align*}
G_{\gamma_p, \hat  H_p}(v)-G_{\gamma, H}(v)
&= \gamma_p \int \frac{t \, d\hat H_p(t)}{1 + tv}- \gamma \int \frac{t \mrd H(t)}{1 + tv}\\
&=\gamma_p \int \frac{t \mrd[\hat H_p-H](t)}{1 + tv} + (\gamma_p-\gamma)\int \frac{t \mrd H(t)}{1 + tv}.
\end{align*}
Both terms in the above display converge to zero uniformly. For the second term, $\gamma_p\to\gamma$, so it is enough to show that the integral is bounded. Now, both $1/v$ and $G(v)$ belong to compact intervals in the range of interest, therefore $\int t/(1 + tv) \mrd H(t) = G(v)+1/v$ also belongs to a compact interval. This shows that the integrand is bounded, and shows that the second term converges to zero uniformly.

 For the first term, $\gamma_p$ is bounded, so we only need to show convergence to zero of the integral. We saw that $\int t/(1 + tv) \mrd H(t)$ is bounded in the range $v\in [l,u]$. By monotonicity, it follows that $t/(1 + tv)$ is also uniformly bounded on that interval, for any $t\in [0,\upper(H)]$. Now, by assumption $\upper(\hat H_p) \to \upper(H)$, thus the measure $\hat H_p- H$ is eventually supported on $t\in [0,(1+\eps)\upper(H)]$. Therefore, by the weak convergence $\hat H_p \tod H$, it follows that  $\int t \mrd[\hat H_p-H](t)/(1 + tv) \to 0$. The convergence is uniform in $v$ by monotonicity. This shows that the second term converges to zero.

This finishes the proof of uniform convergence, and thus that of Theorem~\ref{thm:upper_edge_conv}.

\subsection{Proof of Proposition \ref{prop:ecdf_conv}}
\label{sec:ecdf_conv_pf}

First, $\hat D_j  = \Vert s_j +n^{-1/2} \Phi_j^{1/2} Z_j\Vert^2$ for $j=1,\dots,p$, 
where $n^{-1/2}\Phi_j^{1/2}Z_j$ is the $j$'th column of the noise $N$ and
$s_j$ is the $j$'th column of the signal $S$.
Let $E_j  = \Vert n^{-1/2} \Phi_j^{1/2} Z_j\Vert^2$. 
Then 
$$\max_j \Vert\hat D_j^{1/2} - E_j^{1/2}\Vert \le \max_j \Vert s_j\Vert \to 0.$$
Therefore, to show the claims of the proposition, it is enough to work with $E_j$ instead of $\hat D_j$. The first claim to be proved 
is that $\hat H_p \tod H$, which we may now replace by $\hat G_p\tod H$ for $\hat G_p = p^{-1}\sum_{j=1}^p \delta_{E_j}$. Similarly,
$\upper(\hat H_p) \to \upper(H)$ is equivalent to $\upper(\hat G_p) \to \upper(H)$. 

Let $m = \max_{j} |E_j - \Phi_j|$. 
Proposition \ref{prop:ecdf_conv} inherits 
from Proposition~\ref{prop:noise_op_norm} the condition that $\max_j\Phi_j\le C<\infty$.
Then 
$$m \le \max_j |\Phi_j| \cdot \max_j \bigl|n^{-1}\Vert Z_j\Vert^2-1\bigr| \le C \max_j \bigl|n^{-1}\Vert Z_j\Vert^2-1\bigr| \to 0\quad\mathrm{a.s..}$$
We have also used 
$$\Pr\bigl(n^{-1}\Vert Z_j\Vert^2-1\ge \eps\bigr) 
\le \E\bigl( (n^{-1}\Vert Z_j\Vert^2-1)^{4+2\delta}\bigr)/\eps^{4+2\delta}  = O(n^{-2-\delta}),$$
which follows from the uniformly bounded $8+\eps$-th moments of $Z$. 
Thus 
$$\Pr\bigl(\max_j |n^{-1}\Vert Z_j\Vert^2-1|>\eps\bigr) =O( p \cdot n^{-2-\delta} ),$$ 
which is summable, so a.s.\ convergence follows by the Borel-Cantelli lemma.

Since $m = \max_{j} |E_j - \Phi_j| \to 0$ a.s., and $H_p = p^{-1}\sum_j \delta_{\Phi_j} \tod H$, 
it follows that $\hat G_p = p^{-1}\sum_{j=1}^p \delta_{E_j} \tod H$ and then $\hat H_p\tod H$ too.
Similarly, since  $\upper(H_p) \to \upper(H)$, it follows that $\upper(\hat G_p) \to \upper(H)$
and so $\upper(\hat H_p)\to\upper(H)$. This finishes the proof
of Proposition \ref{prop:ecdf_conv}.

\subsection{Proof of Theorem~\ref{thm:defpa_cons}}
\label{sec:defpa_cons_pf}

We will follow the same broad strategy as in the proof of Theorem \ref{thm:fa_cons}. Namely, we will consider the more general signal-plus-noise models, and establish conditions when deflation works in that setting. Then, we will check those conditions for factor models. 

First, let us introduce some notation.
We write $X_0 = X$, and for $k\ge 0$,  write the deflated matrices as 
\begin{align*}
X_{k+1} 
&= X_{k} -  \sigma_1(X_{k}) u_1(X_{k}) v_1^\top(X_{k})\\
&= X -  \sum_{\ell=1}^k \sigma_\ell(X) u_\ell(X) v_\ell^\top(X).
\end{align*}
We can also write $n^{-1/2}X_{k+1}  = S_{k+1} + N$, where 
$$S_{k+1} = S - n^{-1/2}\sum_{\ell=1}^k \sigma_1(X_{\ell}) u_1(X_{\ell}) v_1^\top(X_{\ell})$$ 
is the deflated signal.
Let us also define $\hat D_k = \diag (n^{-1} X_k^\top X_k) $ and $\hat H_p^k = \esd(\hat D_k)$.

By carefully examining the proof of Theorem \ref{thm:fa_cons},  we can conclude that the consistency results for signal-plus-noise models hold for any fixed $k$ under a simple condition. Let us define the $\ell_2$ column matrix norm $\Vert M\Vert_{\infty,2} = \max_{j} \Vert m_j\Vert_2$, where $m_j$ are the columns of $M$. Then the condition is that the $\ell_2$ column norm of the deflated signals vanishes: $\Vert S_k\Vert_{\infty,2}\to 0$. As a simple self-consistency check, we point out that this agrees with the condition from the non-deflated case,  $\Vert S\Vert_{\infty,2}\to 0$, stated in Proposition \ref{prop:ecdf_conv}.

The decision in the first step is based on checking $\sigma_1(n^{-1/2}X)>\upper(F_{\gamma_p, \hat  H_p})$. For this, the existing analysis implies that DDPA selects the first factor provided it is perceptible. This does not require any additional conditions. 

The decision in the second step is based on $\sigma_1(n^{-1/2}X_1)=\sigma_2(n^{-1/2}X)>\upper(F_{\gamma_p, \hat  H_p^1})$. For this, if we can show that $\Vert S_1\Vert_{\infty,2}\to 0$, then we can conclude that Deflated PA selects the second factor, provided it is perceptible. This is the first nontrivial step.

Recall that $S_1 = S - \sigma_1(X) u_1(X) v_1^\top(X)$, and that we are under the conditions of  Theorem \ref{thm:fa_cons}, when DPA works. Thus we can assume $\Vert S\Vert_{\infty,2}\to 0$, and we need to show that $\Vert \sigma_1(X) u_1(X) v_1^\top(X)\Vert_{\infty,2}\to 0$. Equivalently, $\sigma_1(X) \Vert v_1(X)\Vert_{\infty}\to 0$. 

%

This follows from analyzing spiked covariance models, using a proof strategy pioneered by \cite{benaych2012singular}, but nonetheless with several new steps. Recall that $m_H(z) = \E_H((X-z)^{-1})$ is the Stieltjes transform of the distribution $H$. 

\begin{prop}[Spike $\ell_\infty$ norm] 
\label{prop:spike_norm}
Suppose that $X= \sum_{\ell=1}^k \theta_\ell a_\ell b_\ell^\top + n^{-1/2} Z \Phi^{1/2}$ with $\theta_1>\theta_2>\ldots>\theta_k$ depending on $n,p$ and possibly diverging, and $a_\ell,b_\ell$ unit vectors that are non-random or random and independent of the noise,  under the asssumptions of Proposition \ref{prop:noise_op_norm}. Suppose that
$\theta_\ell \Vert b_\ell\Vert_{\infty}\to  0$ for $\ell=1,\dots,k$.
Suppose moreover that, for any $z\in\mathbb{C}^+$, $x^\top (\Phi -zI_p)^{-1} b_\ell - m_H(z) \cdot x^\top b_\ell \to 0$ uniformly for $\Vert x\Vert \le 1$. 
Then for all $\ell\le k$ such that $\sigma_\ell(X)>b$ a.s., we have $\sigma_\ell(X)\Vert v_\ell(X)\Vert_{\infty}\to  0$ a.s.
\end{prop}
\begin{proof}
See Section~\ref{sec:spike_norm_pf}.
\end{proof}

Continuing with the general theory, the decision in the third step is based on $\sigma_1(n^{-1/2}X_2)=\sigma_3(n^{-1/2}X)>\upper(F_{\gamma_p, \hat  H_p^2})$. For this step, we need to show $\Vert S_2\Vert_{\infty,2}\to 0$. Recall that $S_2 = S_1 - \sigma_1(X_1) u_1(X_1) v_1^\top(X_1)$, and we already showed $\Vert S_1\Vert_{\infty,2}\to 0$. Thus, as above, we need to show that $\sigma_1(X_1)\Vert v_1(X_1)\Vert_{\infty}\to  0$. Now,  $X_1 = X - \sigma_1(X) u_1(X) v_1^\top(X)$, so by definition of the deflation algorithm, $v_2(X) = v_1(X_1)$, $\sigma_2(X) = \sigma_1(X_1)$, and under our conditions, $\sigma_2(X)$  $\Vert v_2(X)\Vert_{\infty}\to 0$ by Proposition \ref{prop:spike_norm}. Thus, this claim follows by definition. Hence, the deflated PA algorithm has the correct behavior in the second step.

The above reasoning becomes an induction argument showing that deflated PA selects all perceptible factors. It remains to argue that it does not select any imperceptible factors. This follows from the previous analysis if $k=0$, so we assume $k \ge 1$.

The decision in the $k+1$-st step is based on checking $\sigma_1(n^{-1/2}X_{k})>\upper(F_{\gamma_p, \hat  H_p^{k+1}})$. Now, $\sigma_1(n^{-1/2}X_{k}) = \sigma_{k+1}(n^{-1/2}X) <b-\eps$ a.s., so for this factor to be non-selected a.s., it is enough to show that $\upper(F_{\gamma_p, \hat  H_p^{k+1}}) \to b$ a.s.. Similarly to the above argument, for this it is enough that $\Vert S_{k}\Vert_{\infty,2}\to 0$, which reduces to $\sigma_1(X_{k-1}) \Vert v_1(X_{k-1})\Vert_{\infty}\to  0$. The key point here is that $v_1(X_{k-1}) = v_{k}(X)$, so this condition follows by Proposition \ref{prop:spike_norm}. Importantly, this singular value satisfies the condition $\sigma_1(X_k)>b$ a.s., so that the above proposition applies. This shows that under the conditions of Proposition \ref{prop:spike_norm}, deflated PA does not select any imperceptible factors a.s.. This finishes the analysis of the signal-plus-noise case.

\subsubsection{Finishing the proof of Theorem \ref{thm:defpa_cons}}

We will check that the conditions of Proposition \ref{prop:spike_norm} hold, using a strategy similar to the proof of Theorem \ref{thm:fa_cons}. In the factor model $n^{-1/2}X = n^{-1/2}U \Psi^{1/2}\Lambda^\top +n^{-1/2}Z\Phi^{1/2}$, the signal component is 
$$S= n^{-1/2} \sum_{\ell=1}^r u_\ell d_\ell^\top = \sum_{\ell=1}^r \theta_\ell a_\ell d_\ell^\top,$$ 
where $a_\ell = u_\ell/\Vert u_\ell\Vert$ and $\theta_\ell = n^{-1/2}\Vert u_\ell\Vert$. 

According to Proposition \ref{prop:spike_norm}, we need the following conditions: 

\begin{compactenum}[\qquad\bf1)]
\item $\theta_\ell \Vert d_\ell\Vert_\infty\to 0$ for $1\le\ell\le k$: Since $u_k$ has IID entries with mean 0 and variance 1, by the LLN $n^{-1/2} \Vert u_k\Vert \to 1$ a.s.. 
Moreover $\Vert d_\ell\Vert\to 0$ by assumption. This guarantees the condition.

\item $\Vert d_\ell\Vert^{-1}(x^\top (\Phi -zI_p)^{-1} d_\ell - m_H(z) \cdot x^\top d_\ell) \to 0$ uniformly in $x$. This follows by assumption.
\end{compactenum}
The noise conditions of Proposition \ref{prop:spike_norm} hold because they are the same assumptions required by  Theorem \ref{thm:fa_cons}. Thus, all conditions are satisfied, finishing the proof. 

\subsection{Proof of Proposition \ref{prop:spike_norm}}
\label{sec:spike_norm_pf}
For simplicity, we handle the case $k=1$ first, so $X= \theta \cdot a b^\top + N$.
Suppose $c$ is a singular value of $X$, with singular vectors $u,v$ of $n$ and $p$ coordinates. Then we have 
$Xv =  cu.$
Denoting $c_1 = a^\top u$, $c_2 = b^\top v$, the equation reads
$( \theta \cdot a b^\top + N) v =  cu$, or $c_2\theta\cdot a + Nv = cu$. By symmetry, $c_1\theta \cdot b+N^\top u = cv$. In matrix form:
\begin{align*}
\begin{bmatrix} 
c I_n & 
-N \\ 
-N^\top & 
c I_p
\end{bmatrix} 
\begin{bmatrix} u \\  v \end{bmatrix}
=
\begin{bmatrix} c_2\theta\cdot a  \\  c_1\theta\cdot b \end{bmatrix}.
\end{align*}
Suppose now that $c$ is not a singular value of $N$. Let $R = (c^2 I_n - N N^\top)^{-1}$, and $\tilde R = (c^2 I_p -  N^\top N)^{-1}$. By the partitioned matrix inverse formula, we obtain
\begin{align*}
\begin{bmatrix} u \\  v \end{bmatrix}
=
\begin{bmatrix} 
c R & 
R N \\ 
N^\top R & 
c \tilde R
\end{bmatrix} 
\begin{bmatrix} c_2\theta\cdot a  \\  c_1\theta\cdot b \end{bmatrix}.
\end{align*}
Thus, 
$$v = c_2\theta\cdot  N^\top Ra +  c_1\theta c \cdot\tilde Rb.$$

Now, we also know that $c>b$ a.s., so that $c$ is a perceptible component. Also, $|c_1|,|c_2| \le 1$, hence to show that $c\Vert v\Vert_\infty \to 0$, it is enough to show that 
$$\theta c \cdot \Vert N^\top Ra\Vert_\infty \to 0\quad\text{and}\quad
\theta c^2 \cdot \Vert \tilde Rb\Vert_\infty \to 0.$$
The two claims are similar, let us focus on the second. We use the representation $\Vert y\Vert_{\infty} = \sup_{\Vert x\Vert_1\le 1} x^\top y$. Thus, $\Vert \tilde Rb\Vert_\infty  = \sup_{\Vert x\Vert_1\le 1} x^\top \tilde Rb$.

Using the \emph{deterministic equivalent} results in \cite{bai2007asymptotics}, and denoting by $m_{\gamma,H}$ the Stieltjes transform of $F_{\gamma,H}$, we have 
\beqs
|x^\top \tilde Rb - x^\top b\cdot m_{\gamma,H}(c^2)|\to 0.
\eeqs
Indeed, this follows from Theorem 1 of \cite{bai2007asymptotics}, taking $v\to 0$, as in \cite{dls2016}. This requires the assumption that $x^\top (\Phi -zI_p)^{-1} b - m_H(z) \cdot x^\top b \to 0$, for $z$ with $\Im(z)>0$ fixed. It is not hard to see that the convergence rate is uniform in $x$, provided that it is uniform in the assumption. 

Hence, to obtain $\theta c^2\cdot \Vert\tilde Rb\Vert_\infty \to 0$, it is enough that $\theta \Vert b\Vert_\infty \cdot c^2 m_{\gamma,H}(c^2) \to 0$. Now, $\theta \Vert b\Vert_\infty \to 0$ by assumption, and  $c^2 m_{\gamma,H}(c^2)$ is uniformly bounded. This shows the required claim. The multispiked case is similar, and thus the proof is omitted.

\subsection{Numerical experiments}

We provide additonal numerical simulations to understand and compare the behavior of our proposed methods.

\subsubsection{DPA versus PA---larger $p$}

We compare DPA to PA, following the setup in the paper, but with larger aspect ratio. Specifically, we take $p = 300$, and $n = 75, 150$. 
The results in Figure \ref{fig:PA2} largely agree with the simulations in the main paper. 

\begin{figure}
\begin{subfigure}{.5\textwidth}
  \centering
  \includegraphics[scale=0.43]{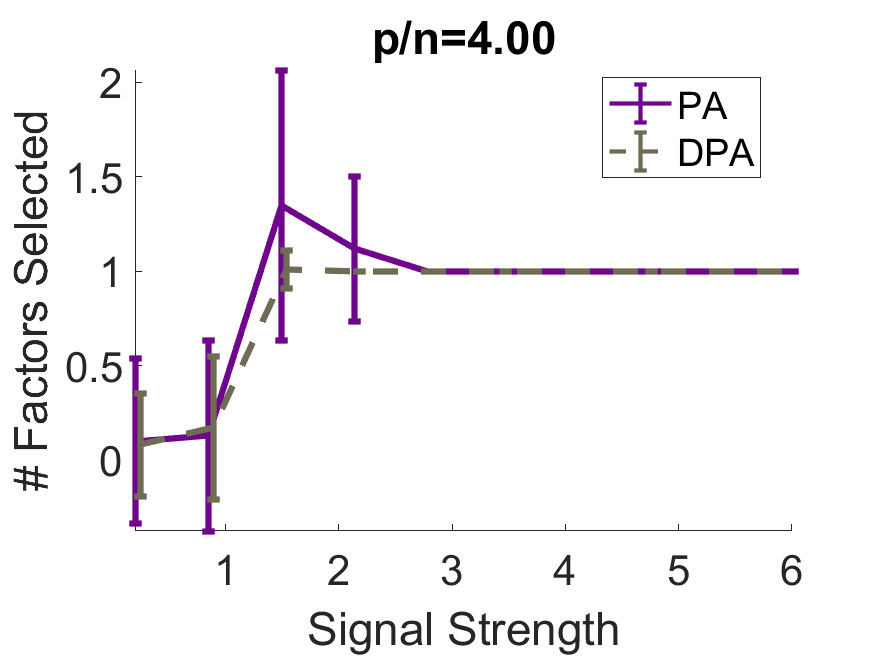}
  \caption{$\gamma = 4$.}
\end{subfigure}
\begin{subfigure}{.5\textwidth}
  \centering
  \includegraphics[scale=0.43]{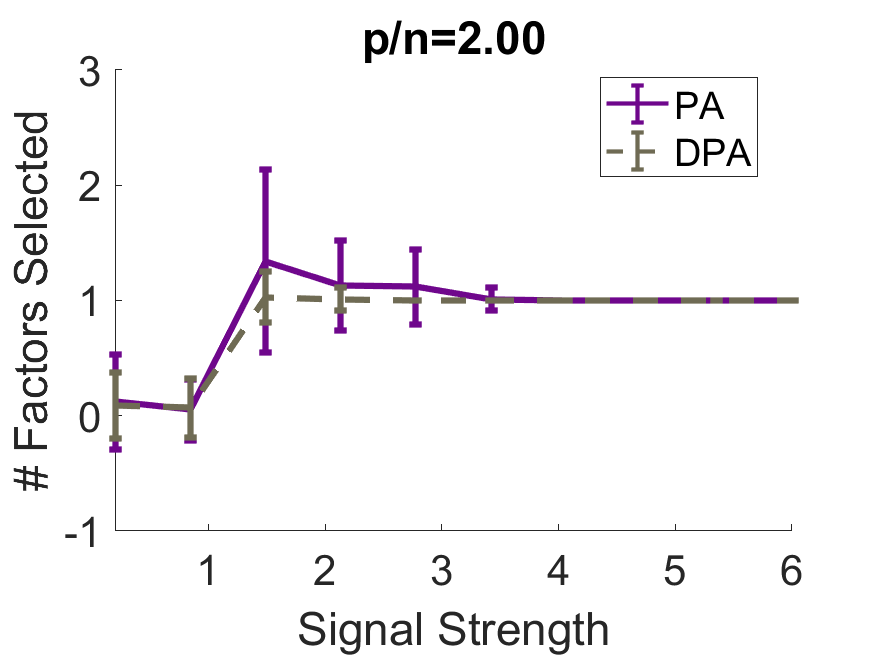}
  \caption{$\gamma = 2$.}
\end{subfigure}
\caption{ Mean and $\pm1$SD of number of factors selected by PA and  DPA as a function of signal strength.}
\label{fig:PA2}
\end{figure}

\subsubsection{DPA versus PA---non-Gaussian data}

We compare DPA to PA, following the setup in the paper, but with non-Gaussian data. Specifically, we take $p = 300$, and $n = 75$. We generate the noise $\ep$ as iid with Bernoulli($s$) entries. We standardize each entry by subtracting $s$ and dividing by $[s(1-s)]^{1/2}$. We take $s = 1/2$ and $1/20$. 
The results in Figure \ref{fig:PA3} largely agree with the simulations in the main paper.  This is expected from universality results in random matrix theory. 

\begin{figure}
\begin{subfigure}{.5\textwidth}
  \centering
  \includegraphics[scale=0.43]{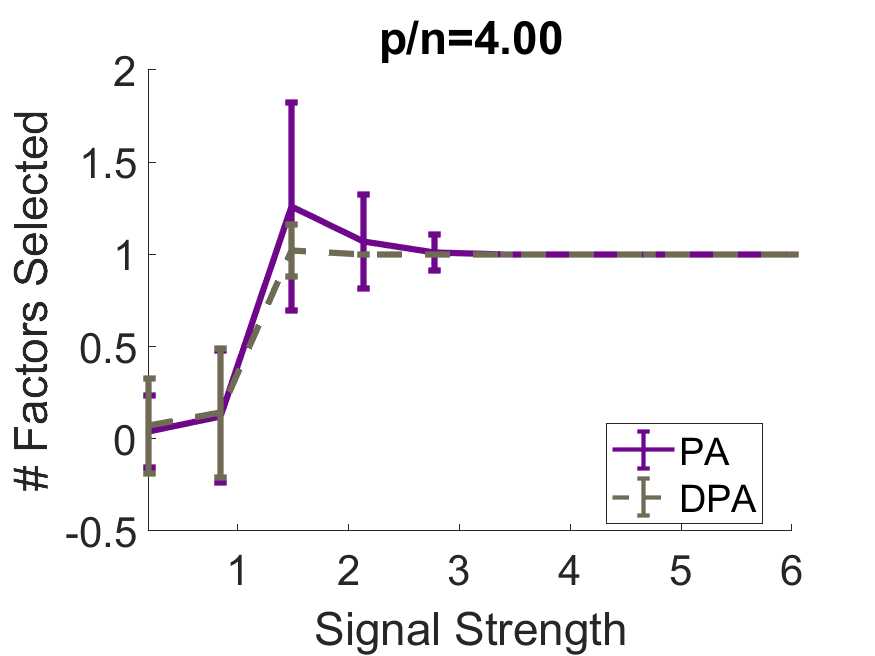}
  \caption{$s = 1/2$}
\end{subfigure}
\begin{subfigure}{.5\textwidth}
  \centering
  \includegraphics[scale=0.43]{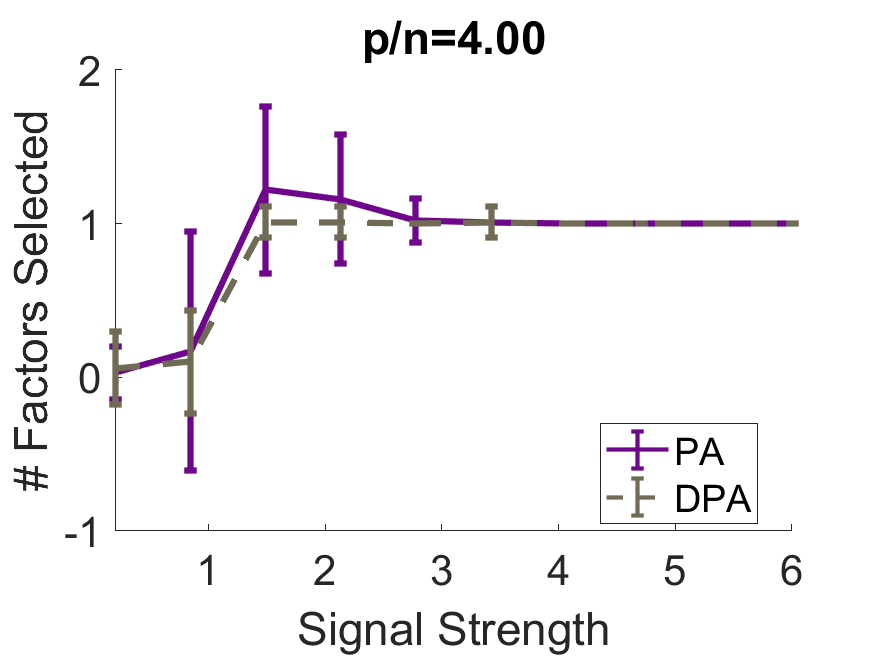}
  \caption{$s = 1/10$}
\end{subfigure}
\caption{ Mean and $\pm1$SD of number of factors selected by PA and  DPA as a function of signal strength for non-Gaussian data.}
\label{fig:PA3}
\end{figure}

\subsubsection{DDPA Clustering}

We check the behavior of DDPA under clustered data. We want to simulate hierarchically clustered data that mimics the structure of HGDP: there are several continents (Europe, Asia), and within each continent there are several regions (Western Europe, Eastern Europe), and within each region there are several countries (e.g., within Western Europe there is Great Britain, Germany, etc). To achieve this, we let the covariance matrix of the $n$ samples  $\Gamma$ be a depth-$d$ BinaryTree covariance structure, used in genetics to model the correlations between populations with an evolutionary history described by a balanced binary tree \citep{pickrell2012inference}. The eigenvalue spectrum of $\Gamma$ equals $H_n = \sum_{i=1}^{d} 2^{-i} \delta_{2^i}$, where $\delta_c$ is the point mass at $c$. This example was also used in \cite{dobriban2015high}.

\begin{figure}
  \centering
  \includegraphics[scale=0.75]{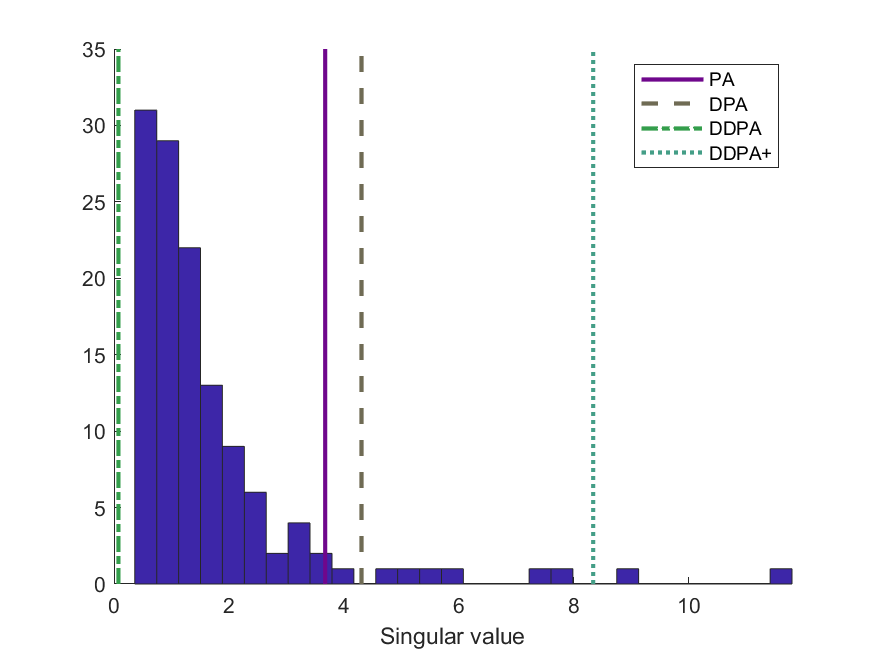}
\caption{Singular value histogram of simulated clustered data. Overlaid  are the thresholds where selection of factors stops for the four methods.}
\label{fig:h}
\end{figure}

We choose $d = 7$, so $n = 127$, and $p = 254$. We simulate from the same 1-factor model as before, except that the noise has the form $\Gamma^{1/2} \ep \Sigma^{1/2}$. Here $\ep$ is normal, so that for $\Gamma$ we can choose a diagonal matrix with the right eigenvalues (given above). 

On Figure \ref{fig:h} we see plot the same plot as for the HGDP data, but for the simulated data. We observe the following

\benum
\item The overall shape of the eigenvalue distribution is similar to that of the HGDP data. There is a long right tailed bulk, and a few separated eigenvalues. However, they do not correspond to real factors, but are instead due to the clustering.
\item The overall performance of the four methods is similar to that on the HGDP data. DDPA+ is the most conservative, selecting only two factors. Then, in order, DPA selects somewhat more, PA even more, and DDPA selects all factors. Thus, the failure of DDPA is also captured on this dataset. 
\eenum 

In conclusion, the behavior of this simulation is quite similar to the HGDP data. 
This is consistent with an explanation of the DDPA results there being
due to a hierarchical clustering structure within those populations.

\subsubsection{PA vs DPA on clustered data}

To address a question from a referee, we perform additional simulations comparing PA to DPA on clustered data. This is a setting where the two algorithms are not designed to perform well. Indeed, any algorithm designed for this setting should take the clustering structure into account. However, it is still interesting and valuable to understand the behavior of our methods in this setting, with the understanding that they are not supposed to be ``correct". With this in mind, we try to understand the following phenomena (for which questions we thank a referee): 

\begin{enumerate}
\item PA has been reported in several works to overestimate the number of factors. However, we have found that PA can underestimate the number of factors due to shadowing. Is there a cancellation?
\item Does the same hold for DPA?
\item How much does the sample size matter?
\end{enumerate}

To study these issues, we perform simulations in the setting of clustered data introduced above, where we change the strength of the large signal, and also change the dimension. Note that, in the model above, the sample size is fixed as a power of two, and so it is easier to study the effects of changing the sample size by changing the dimension instead. Specifically, we change the aspect ratio from $\gamma = 2$ to $\gamma = 1.5$. We also vary the signal strength on a grid from 0 to 10. 

\begin{figure}
\begin{subfigure}{.5\textwidth}
  \centering
  \includegraphics[scale=0.43]{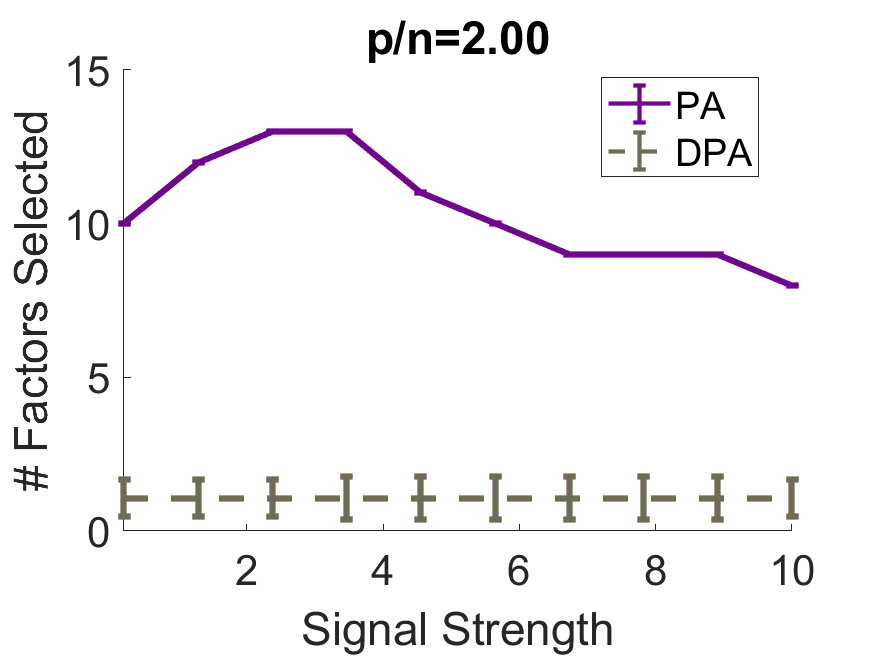}
  \caption{$\gamma = 2$}
\end{subfigure}
\begin{subfigure}{.5\textwidth}
  \centering
  \includegraphics[scale=0.43]{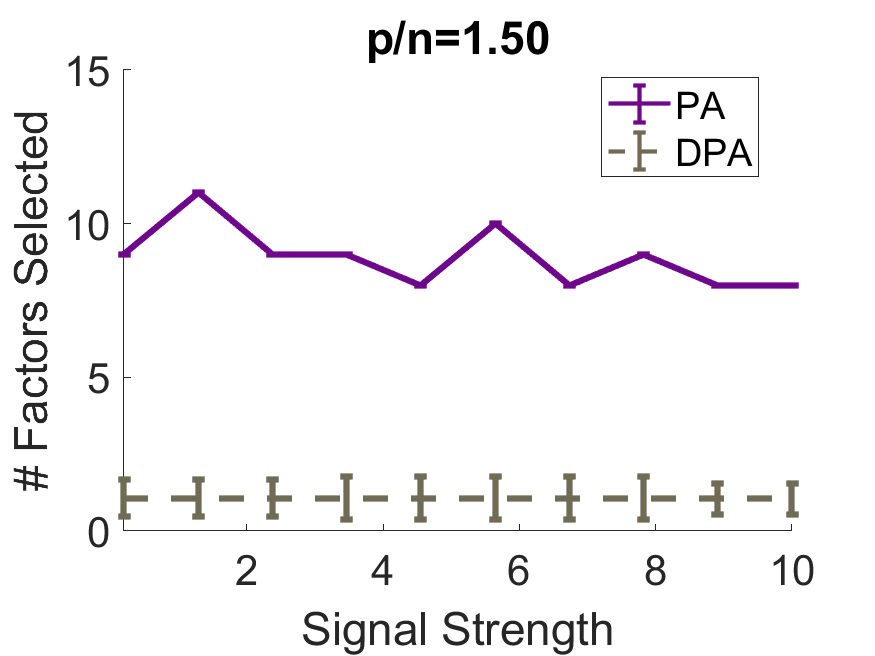}
  \caption{$\gamma = 1.5$}
\end{subfigure}
\caption{Number of factors selected by PA and  DPA as a function of signal strength on clustered data.}
\label{fig:p}
\end{figure}

The results in Figure \ref{fig:p} suggest the following answers to the questions above: 

\begin{enumerate}
\item Is there a cancellation? Yes, there appears to be a cancellation. For a small signal strength, the number of signals selected by PA can be very large, but for larger signal strength, this number decreases.
\item Does the same hold for DPA? First of all, DPA is more accurate, selecting close to one factor on average. It is not clear if the same cancellation phenomenon holds for DPA. 
\item How much does the sample size matter? A larger sample size does indeed appear to help, at least to some extent. On the right plot, we see that the number of factors selected by PA is smaller (which is more accurate in the present case). However, the number of factors selected appears to plateau at a large number.
\end{enumerate}

{
\setlength{\bibsep}{0.2pt plus 0.3ex}
\bibliographystyle{plainnat-abbrev}
\bibliography{references}
}

\end{document}